\documentclass[11pt]{article}
\usepackage[dvips]{graphicx}
\usepackage{amsmath,amsthm,amscd}
\usepackage{amssymb}
\usepackage{fancybox}
\usepackage{latexsym}
\usepackage[dvips]{color}

\def\KL{\mathrm{KL}}

\def\Real{\Rbb}
\def\Rbb{\mathbb{R}}

\def\<{\langle}
\def\>{\rangle}

\def\subto{\text{\rm subject to }}

\newcommand{\tr}[1]{{\mathrm{tr}(#1)}}

\newtheorem{theorem}{Theorem}
\newtheorem{lemma}[theorem]{Lemma}
\newtheorem{proposition}[theorem]{Proposition}

\newtheorem{assumption}{Assumption}
\newtheorem{definition}{Definition}
\newtheorem{example}{Example}

\def\0{\mbox{\bf 0}}
\def\1{\mbox{\bf 1}}
\def\2{\mbox{\bf 2}}
\def\3{\mbox{\bf 3}}
\def\4{\mbox{\bf 4}}
\def\5{\mbox{\bf 5}}
\def\6{\mbox{\bf 6}}
\def\7{\mbox{\bf 7}}
\def\8{\mbox{\bf 8}}
\def\9{\mbox{\bf 9}}

\definecolor{cyan}{cmyk}{1,0,0,0}
\definecolor{lightcyan}{cmyk}{0.5,0,0,0}
\definecolor{pastelcyan}{cmyk}{0.25,0,0,0}
\definecolor{magenta}{cmyk}{0,1,0,0}
\definecolor{yellow}{cmyk}{0,0,1,0}
\definecolor{lightyellow}{cmyk}{0,0,0.5,0}
\definecolor{pastelyellow}{cmyk}{0,0,0.25,0}
\definecolor{black}{cmyk}{0,0,0,1}
\definecolor{darkgray}{cmyk}{0,0,0,0.75}
\definecolor{gray}{cmyk}{0,0,0,0.5}
\definecolor{lightgray}{cmyk}{0,0,0,0.25}
\definecolor{white}{cmyk}{0,0,0,0}
\definecolor{red}{cmyk}{0,1,1,0}
\definecolor{orange}{cmyk}{0,0.5,1,0}
\definecolor{scarlet}{cmyk}{0,1,0.5,0}
\definecolor{brown}{cmyk}{0.5,0.75,1,0}
\definecolor{camel}{cmyk}{0.25,0.375,0.5,0}
\definecolor{cream}{cmyk}{0,0.2,0.3,0}
\definecolor{green}{cmyk}{1,0,1,0}
\definecolor{lightgreen}{cmyk}{0.5,0,0.5,0}
\definecolor{pastelgreen}{cmyk}{0.25,0,0.25,0}
\definecolor{mossgreen}{cmyk}{0.64,0.4,1,0}
\definecolor{yellowgreen}{cmyk}{0.5,0,1,0}
\definecolor{skyblue}{cmyk}{0.4,0.16,0,0}
\definecolor{royal}{cmyk}{1.0,0.5,0,0}
\definecolor{navyblue}{cmyk}{0.9,0.75,0.5,0}
\definecolor{lightnavy}{cmyk}{0.4,0.3,0.2,0}
\definecolor{blue}{cmyk}{1,1,0,0}
\definecolor{lightblue}{cmyk}{0.5,0.5,0,0}
\definecolor{lavender}{cmyk}{0.25,0.25,0,0}
\definecolor{violet}{cmyk}{0.75,1,0.25,0}
\definecolor{purple}{cmyk}{0.5,1,0.5,0}
\definecolor{pink}{cmyk}{0,0.5,0,0}
\definecolor{pastelpink}{cmyk}{0,0.25,0,0}

\def\PD{\mathrm{PD}}

\title{A Bregman Extension of quasi-Newton updates II:\\ Convergence and Robustness Properties}
\author{
  Takafumi Kanamori\\ Nagoya University \\ \tt{kanamori@is.nagoya-u.ac.jp}
  \and
  Atsumi Ohara\\ Osaka University\\ \tt{ohara@sys.es.osaka-u.ac.jp}
 }
\date{}

\begin{document}
\maketitle
\begin{abstract}
 We propose an extension of quasi-Newton methods, and investigate the convergence and the
 robustness properties of the proposed update formulae for the approximate Hessian
 matrix. 
 Fletcher has studied a variational problem which derives the approximate Hessian update
 formula of the quasi-Newton methods. 
 We point out that the variational problem is identical to optimization of the
 Kullback-Leibler divergence, which is a discrepancy measure 
 between two probability distributions. 
 Then, we introduce the Bregman divergence as an extension of the Kullback-Leibler
 divergence, and derive extended quasi-Newton update formulae based on the variational
 problem with the Bregman divergence. 
 The proposed update formulae belong to a class of self-scaling quasi-Newton methods. 
 We study the convergence property of the proposed quasi-Newton method, and 
 moreover, we apply the tools in the robust statistics to analyze the robustness
 property of the Hessian update formulae against the numerical rounding errors included in
 the line search for the step length. 
 As the result, we found that the influence of the inexact line search is bounded only for
 the standard BFGS formula for the Hessian approximation. 
 Numerical studies are conducted to verify the usefulness of the tools borrowed from
 robust statistics. 
\end{abstract}

\section{Introduction}
\label{sec:Introduction}
We consider quasi-Newton methods for the unconstrained optimization problem 
\begin{align}
 \label{eqn:main_opt_problem}
 \text{minimize}\  f(x),\quad x\in\Real^n, 
\end{align}
in which the function $f:\Real^n\rightarrow\Real$ is twice continuously
differentiable on $\Real^n$. 
The quasi-Newton method is known to be one of the most successful methods for
unconstrained function minimization. 
Details are shown in \cite{nocedal99:_numer_optim,luenberger08:_linear_and_nonlin_progr}
and references therein. 

The main purpose of this paper is to present an extended framework of quasi-Newton method,
and to study the robustness property of quasi-Newton update formulae against numerical
errors of line search. 
There are mainly two standard quasi-Newton method; one
is the DFP formula and the other is the BFGS formula. 
Fletcher \cite{fletcher91:_new_resul_for_quasi_newton_formul} has pointed out that the 
standard formulae, DFP and BFGS, are obtained as the optimal solution of 
a variational problem over the set of positive definite matrices. 
Along this line, we extend the quasi-Newton update formula. 
Then, we study the robustness property of the extended quasi-Newton methods, where 
we apply some techniques exploited in the field of robust statistics
\cite{Hampel_etal86}. 

We briefly introduce quasi-Newton formulae and its variational result. 
In quasi-Newton method, a sequence $\{x_k\}_{k=0}^{\infty}\subset\Real^n$ is successively
generated in a manner such that $x_{k+1}=x_k-\alpha_k B_k^{-1}\nabla f(x_k)$. 
The coefficient $\alpha_k\in\Real$ is a step-size computed by a line search, 
and $B_k$ is a positive definite matrix approximating the Hessian matrix 
$\nabla^2 f(x_k)$ at the point $x_k$. 
Let $s_k$ and $y_k$ be column vectors defined by
\begin{align*}
 s_k=x_{k+1}-x_k=-\alpha_kB_k^{-1}\nabla f(x_k),
 \qquad y_k=\nabla f(x_{k+1})-\nabla f(x_k). 
\end{align*}
We need a Hessian approximation $B_{k+1}$ for $\nabla^2 f(x_{k+1})$ 
to keep on the computation. 
In the DFP method, $B_{k+1}$ is given by 
\begin{align}
 \label{eqn:DFP-update-formula}
 B_{k+1}
 &=
 B^{DFP}[B_k;s_k,y_k] 
 := B_k-\frac{B_ks_ky_k^\top+y_ks_k^\top B_k}{s_k^\top y_k}
 +s_k^\top B_ks_k \frac{y_k y_k^\top}{(s_k^\top y_k)^2}+\frac{y_ky_k^\top}{s_k^\top y_k}, 
\end{align}
and the BFGS method provides the different formula such that 
\begin{align}
 \label{eqn:BFGS-update-formula}
 B_{k+1}
 &= B^{BFGS}[B_k;s_k,y_k]
 :=B_k-\frac{B_ks_ks_k^\top B_k}{s_k^\top B_ks_k}+\frac{y_ky_k^\top}{s_k^\top y_k}, 
\end{align}
When $B_k\in\mathrm{PD}(n)$ and $s_k^\top y_k>0$ hold, 
both $B^{DFP}[B_k;s_k,y_k]$ and $B^{BFGS}[B_k;s_k,y_k]$ are also positive definite
matrices. 
In practice, the Cholesky decomposition of $B_k$ will be
successively updated in order to compute the search direction $-B_k^{-1}\nabla f(x_k)$
efficiently. The idea of updating Cholesky factors is pioneered by Gill and Murray
\cite{gill72:_quasi_newton_method_for_uncon_optim}. 
Note that the equality 
\begin{align*}
 B^{DFP}[B_k;s_k,y_k]^{-1}=B^{BFGS}[B_k^{-1};y_k,s_k]
\end{align*}
holds. 
Hence, the update formula for the inverse $H_{k+1}=B_{k+1}^{-1}$ can be directly derived 
from $H_k=B_k^{-1}$ without computing inversion of matrix. 



We introduce a variational approach in quasi-Newton methods. 
Let $\PD(n)$ be the set of all $n$ by $n$ symmetric positive definite matrices, and 
the function $\psi:\mathrm{PD}(n)\rightarrow\Real$ be a strictly convex function 
over $\PD(n)$ defined by 
\begin{align*}
 \psi(A)=\tr{A}-\log\det{A}. 
\end{align*}
Fletcher \cite{fletcher91:_new_resul_for_quasi_newton_formul} has shown that the DFP
update formula \eqref{eqn:DFP-update-formula} is obtained as the unique solution 
of the constraint optimization problem, 
\begin{align*}
 \min_{B\in\mathrm{PD}(n)}\ \psi(B_k^{1/2}B^{-1}B_k^{1/2})\quad \subto\ Bs_k=y_k, 
\end{align*}
where $A^{1/2}$ for $A\in\PD(n)$ is the matrix satisfying $A^{1/2}\in\PD(n)$
and $(A^{1/2})^2=A$. 
The BFGS formula is also obtained as the optimal solution of 
\begin{align*}
 \min_{B\in\mathrm{PD}(n)}\ \psi(B_k^{-1/2}BB_k^{-1/2})\quad \subto\ Bs_k=y_k, 
\end{align*}
in which $B_k^{-1/2}$ denotes $(B_k^{-1})^{1/2}$ or equivalently $(B_k^{1/2})^{-1}$. 

It will be worthwhile to point out that the function $\psi$ is identical to
Kullback-Leibler(KL) divergence \cite{AmariNagaoka00,kullback51:_infor_and_suffic} 
up to an additive constant. 
Let $N_n(0,P)$ be the $n$ dimensional Gaussian distribution with mean zero and
variance-covariance matrix $P\in\mathrm{PD}(n)$, then the KL-divergence between $N_n(0,P)$
and $N_n(0,Q)$ is defined by 
\begin{align*}
 \mathrm{KL}(P,Q)
 &=\tr{PQ^{-1}}-\log\det(PQ^{-1})-n
\end{align*}
which is equal to $\psi(Q^{-1/2}PQ^{-1/2})-n$. 
The KL-divergence is regarded as a generalization of squared  distance over the space of
probability distributions. 
Using the KL-divergence, we can represent the update formulas as the optimal solution of
the following minimization problems, 
\begin{align}
 \text{(DFP)}\qquad&\min_{B\in\mathrm{PD}(n)}\  \mathrm{KL}(B_k,B)\quad \subto\  Bs_k=y_k,
 \label{eqn:DFP}\\
 \text{(BFGS)}\qquad&\min_{B\in\mathrm{PD}(n)}\ \mathrm{KL}(B,B_k)\quad \subto\  Bs_k=y_k. 
 \label{eqn:BFGS}
\end{align} 
The KL-divergence is asymmetric, that is, $\KL(P,Q)\neq \KL(Q,P)$ in general. 
Hence the above problems will provide different solutions. 

Here is the brief outline of the article. 
In Section \ref{sec:Bregman_Div} we introduce the so-called Bregman divergence 
which is an extension of the KL-divergence. 
In Section \ref{sec:extension_quasi-Newton_updates}, 
an extended quasi-Newton formula is derived based on the Bregman divergence. 
In Section \ref{sec:Convergence_Analysis}, the convergence property of the proposed
quasi-Newton method is studied, and 
Section \ref{sec:Robustness} is devoted to discuss the robustness of the Hessian update
formula. 
Numerical simulations are presented in Section \ref{sec:Numerical_Studies}. 
We conclude with a discussion and outlook in Section \ref{sec:Concluding_Remarks}. 
Some proofs of the theorems are postponed to Appendix. 

Throughout the paper, we use the following notations: 
The set of positive real numbers are denoted as $\Real_+\subset\Real$. 
Let $\det{A}$ be the determinant of square matrix $A$, and 
$\mathrm{GL}(n)$ denotes the set of $n$ by $n$ non-degenerate real matrices. 
The set of all $n$ by $n$ real symmetric matrices is denoted as $\mathrm{Sym}(n)$, and 
let $\mathrm{PD}(n)\subset\mathrm{GL}(n)\cap\mathrm{Sym}(n)$ be the set of $n$ by $n$
symmetric positive definite matrices. 
For two square matrices $A,\,B$, the inner product 
$\<A,B\>$ is defined by $\tr{A B^\top}$, and $\|A\|_F$ is the Frobenius norm defined by
the square root of $\<A,A\>$. 
Throughout the paper we only deal with the inner product of symmetric matrices, and the 
transposition in the trace will be dropped. For a vector $x$, $\|x\|$ denotes the
Euclidean norm. The first and second order derivative of a function
$f:\Real\rightarrow\Real$ are denoted as $f'$ and $f''$, respectively.

\section{Bregman Divergence induced from Potential Functions}
\label{sec:Bregman_Div}
As introduced in Section \ref{sec:Introduction}, 
the update formulae of the DFP and the BFGS methods are 
derived from the optimization problem of KL-divergence. 
In this section we introduce Bregman divergence 
\cite{bregman67:_relax_method_of_findin_common} which is an extension of the
KL-divergence. Especially we focus on the Bregman divergence induced from potential 
function. Then, we present extended quasi-Newton formulae derived from the variational
problem for the Bregman divergence. 

Let $\varphi:\PD(n)\rightarrow\Real$ be a differentiable, strictly convex function that
maps positive definite matrices to real numbers. 
We define {\em Bregman divergence} of the matrix $P$ from the matrix $Q$ as 
\begin{align}
 D(P,Q)=\varphi(P)-\varphi(Q)-\<\nabla\varphi(Q),P-Q\>, 
 \label{eqn:def_bregman_div}
\end{align}
where $\nabla\varphi(Q)$ is the $n$ by $n$ matrix whose $(i,j)$ element is given as
$\frac{\partial\varphi}{\partial Q_{ij}}(Q)$. 
The strict convexity of $\varphi$ guarantees that $D(P,Q)$ is non-negative and equals to
zero if and only if $P=Q$ holds. 
Figure \ref{fig:bregman-div} illustrates the relation between the function 
$\varphi$ and the Bregman divergence. 
Note that $D(P,Q)$ is convex in $P$ but not necessarily
convex in $Q$. 
Bregman divergences have been well studied for nearness problems in the fields of
statistics and machine learning 
\cite{banerjee05:_clust_with_bregm_diver, 
dhillon07:_matrix_nearn_probl_with_bregm_diver,
murata04:_infor_geomet_u_boost_bregm_diver}. 
\begin{figure}
 \begin{center}
\includegraphics[scale=0.5]{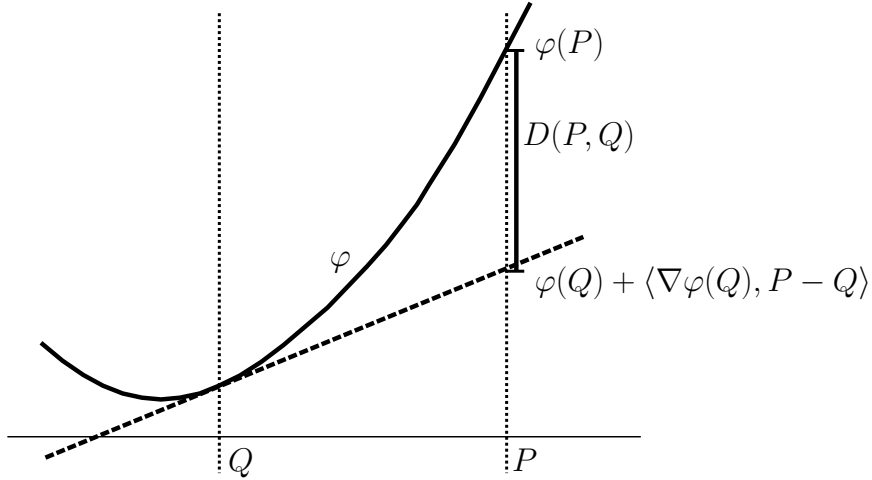}  
  \caption{The Bregman divergence defined by the strictly convex function
  $\varphi:\PD(n)\rightarrow\Real$. 
  Due to the strict convexity of $\varphi$, the function $\varphi(P)$ lies above 
  its tangents $\varphi(Q)+\<\nabla\varphi(Q),P-Q\>$. 
  Hence the non-negativity of the Bregman divergence $D(P,Q)$ is guaranteed. }
  \label{fig:bregman-div}
 \end{center}
\end{figure}

In this paper, we focus on the Bregman divergence induced from potential function
\cite{ohara05:_geomet_posit_defin_matric_and}. 
Let $V:\Real_{+}\rightarrow\Real$ be a strictly convex, decreasing, and third order 
continuously differentiable function. For the derivative $V'$, the inequality $V'<0$ holds
from the assumption. Indeed, the assumption leads to $V'\leq 0$ and 
$V''\geq0$, and if $V'(z_0)=0$ holds for some $z_0\in\Real_+$, then $V'(z)=0$ holds for
 all $z\geq z_0$. 
 Hence $V$ is affine function for $z\geq z_0$. This contradicts the strict convexity of
 $V$. We define the functions $\nu_V:\Real_+\rightarrow\Real$ and
 $\beta_V:\Real_+\rightarrow\Real$ such that  
 \begin{align*}
  \nu_V(z)=-z V'(z),\qquad 
  \beta_V(z)=\frac{z\nu_V'(z)}{\nu_V(z)}. 
 \end{align*}
 The subscript $V$ of $\nu_V$ and $\beta_V$ will be dropped if there is no confusion. 
\begin{definition}[potential function]
 Let $V:\Real_{+}\rightarrow\Real$ be a function which is strictly convex, decreasing, and
 third order continuously differentiable. 
 Suppose that the functions $\nu$ and $\beta$ defined from $V$ satisfy the following conditions:
 \begin{align}
  \label{eqn:nu-condition}
  \nu(z)&>0,\\
  \label{eqn:beta-condition}
  \beta(z)&<\frac{1}{n}
 \end{align}
 for all $z>0$ and 
 \begin{align}
  \label{eqn:nu-limit-condition}
  \lim_{z\rightarrow+0}\frac{z}{\nu(z)^{n-1}}=0. 
 \end{align}
 Then, $V$ is called potential function or potential for short. 
 For $P\in\mathrm{PD}(n)$, the function $V(\det{P})$ is also referred to as potential on
 $\mathrm{PD}(n)$. 
\end{definition}
As shown in \cite{ohara05:_geomet_posit_defin_matric_and}, the function $V(\det{P})$ is
strictly convex in $P\in\mathrm{PD}(n)$ if and only if $V$ satisfies
\eqref{eqn:nu-condition} and  \eqref{eqn:beta-condition}. 
The condition \eqref{eqn:nu-limit-condition} guarantees the existence of Hessian update
formula, which is discussed in Section \ref{sec:extension_quasi-Newton_updates}. 

Given a potential function $V$, the Bregman divergence defined from the potential function
$\varphi(P)=V(\det P)$ in \eqref{eqn:def_bregman_div} is denotes as $D_V(P,Q)$, and  
referred to as {\em $V$-Bregman divergence}. 
The $V$-Bregman divergence has the form of 
\begin{align*}
 D_V(P,Q) = V(\det{P})-V(\det{Q})+\nu(\det{Q})\<Q^{-1},P\>-n\nu(\det{Q}). 
\end{align*}
Indeed, substituting 
 \begin{align*}
  (\nabla V(\det{Q}))_{ij}
  = \frac{dV(\det{Q})}{dQ_{ij}}
  = V'(\det{Q})\frac{d\det{Q}}{dQ_{ij}}
  =-\nu(\det{Q})(Q^{-1})_{ij}, 
 \end{align*}
into \eqref{eqn:def_bregman_div}, we obtain the expression of $D_V(P,Q)$. 
Below we show some examples of $V$-Bregman divergence. 
\begin{example}
 \label{example:KL-div}
For the negative logarithmic function $V(z)=-\log(z)$, we have
$\nu(z)=1$. 
Then $V$-divergence is equal to KL-divergence, 
\begin{align*}
 D_V(P,Q)=\mathrm{KL}(P,Q)=\<P,Q^{-1}\>-\log\det(PQ^{-1})-n. 
\end{align*}
 Note that $\KL(P,Q)=\KL(Q^{-1},P^{-1})$ holds. Hence, $\KL(P,Q)$ is convex in both $P$
 and $Q^{-1}$. 
\end{example}

\begin{example}
 \label{example:power-div}
 For the power potential $V(z)=(1-z^\gamma)/\gamma$ with $\gamma<1/n$, 
 we have $\nu(z)=z^\gamma$ and $\beta(z)=\gamma$. Then, we obtain
 \begin{align*}
  D_V(P,Q)=
  (\det{Q})^\gamma
  \bigg\{\<P,Q^{-1}\>+\frac{1-(\det{PQ^{-1}})^\gamma}{\gamma}-n\bigg\}. 
 \end{align*}
 The KL-divergence is recovered by taking the limit of $\gamma\rightarrow0$. 
\end{example}

\begin{example}
 \label{example:bounded-div}
 For $0\leq a<b$, let $V(z)$ be $V(z)=a\log(a z+1)-b\log(z)$. Then
 $V(z)$ is a convex and decreasing function, and we obtain 
 \begin{align*}
  \nu(z)=b-a+\frac{a}{az+1}>0,\qquad 
  \beta(z)=\frac{-a^2z}{(az+1)(a(b-a)z+b)}\leq 0
 \end{align*}
 for $z>0$. 
 The negative-log potential is derived by setting $a=0,\,b=1$. This potential satisfies
 the inequality $0<b-a\leq \nu(z)\leq b$. The bounding condition of $\nu$ will be assumed
 in the convergence analysis of Section \ref{sec:Convergence_Analysis}. 
\end{example}
We apply $V$-Bregman divergences to extend quasi-Newton update formula.

\section{Extended quasi-Newton update formula}
\label{sec:extension_quasi-Newton_updates}
To extend the standard quasi-Newton methods, we consider the optimization problem of the
$V$-Bregman divergence instead of the KL-divergence. 
Let us define the $V$-BFGS formula as the optimal solution of the problem, 
\begin{align}
 \label{eqn:V-BFGS-hessian-update-prob}
 \text{($V$-BFGS)}\qquad 
 \min_{B\in\mathrm{PD}(n)}\ D_V(B,B_k),\quad \text{subject to}\ \ Bs_k=y_k. 
\end{align}
Next we define $V$-DFP update formula which is an extension of the standard DFP formula
\eqref{eqn:DFP-update-formula}. Note that KL-divergence satisfies
$\KL(P,Q)=\KL(Q^{-1},P^{-1})$. 

Then, the optimization problem associated with the DFP update formula \eqref{eqn:DFP} 
can be extended to the problem, 
\begin{align}
 \label{eqn:V-DFP-hessian-update-prob}
 \text{($V$-DFP)}\qquad 
 \min_{B\in\mathrm{PD}(n)}\ D_V(B^{-1},B_k^{-1}),\quad \text{subject to}\ \ Bs_k=y_k. 
\end{align}
The problem \eqref{eqn:V-DFP-hessian-update-prob} is convex in $B^{-1}$, since
the objective function $D_V(B^{-1},B_k^{-1})$ is convex in $B^{-1}$ and 
the constraint $s_k=B^{-1}y_k$ is affine in $B^{-1}$. 
Mainly we consider the $V$-BFGS update formula. 
The argument on the $V$-DFP update is almost the same. 
\begin{theorem}
 \label{theorem:V-BFGS-form}
 Let $B_k\in\mathrm{PD}(n)$, and suppose $s_k^\top y_k>0$. 
 Then the problem \eqref{eqn:V-BFGS-hessian-update-prob} has the unique 
 optimal solution $B_{k+1}\in\mathrm{PD}(n)$ satisfying 
\begin{align}
 B_{k+1} &=
\frac{\nu(\det{B_{k+1}})}{\nu(\det{B_k})}B^{BFGS}[B_k;s_k,y_k]
 +\bigg(1-\frac{\nu(\det{B_{k+1}})}{\nu(\det{B_k})}\bigg) 
 \frac{y_ky_k^\top}{s_k^\top y_k}. 
 \label{eqn:update-formula-V-BFGS}
\end{align}
\end{theorem}
The proof is found in Appendix \ref{appendix:VBFGS-formula}. 

Note that the $V$-BFGS update formula is represented by the affine sum of
$B^{BFGS}[B_k;s_k,y_k]$ and $y_ky_k^\top/s_k^\top y_k$. 
This form is equivalent to the self-scaling quasi-Newton update 
\cite{oren74:_self_scalin_variab_metric_ssvm,nocedal93:_analy_of_self_scalin_quasi_newton_method}
defined as 
\begin{align}
 \label{eqn:self-scaling}
 B_{k+1} &=\theta_k B^{BFGS}[B_k;s_k,y_k]
 +(1-\theta_k) \frac{y_ky_k^\top}{s_k^\top y_k}, 
\end{align}
where $\theta_k$ is a positive real number. 
In the $V$-BFGS update formula, the coefficient 
$\theta_k$ is determined from the function $\nu$. 
The inverse of the matrix \eqref{eqn:self-scaling} is given by 
\begin{align}
\label{eqn:self-scaling-inverse}
 B_{k+1}^{-1} &=
 \frac{1}{\theta_k} (B^{BFGS}[B_k;s_k,y_k])^{-1} +\bigg(1-\frac{1}{\theta_k}\bigg)
 \frac{s_ks_k^\top}{s_k^\top y_k}. 
\end{align}
As the result, for any $\theta_k> 0$, the matrix $B_{k+1}$
in \eqref{eqn:self-scaling} is positive definite. Indeed, for 
$0<\theta_k\leq 1$ the expression \eqref{eqn:self-scaling} guarantees the positive
definiteness of $B_{k+1}$, and for $1<\theta_k$, the expression
\eqref{eqn:self-scaling-inverse} implies $B_{k+1}\in\mathrm{PD}(n)$. 
Therefore $B_{k+1}$ in \eqref{eqn:update-formula-V-BFGS} is also positive definite
matrix, since any potential $V$ satisfies $\nu_V>0$. 

In the self-scaling update formula in \eqref{eqn:self-scaling}, the choice 
\begin{align}
 \theta_k=\frac{s_k^\top y_k}{s_k^\top B_ks_k}
 \label{eqn:popular-self-scaling-parameter}
\end{align}
is often recommended. 
As analyzed in \cite{nocedal93:_analy_of_self_scalin_quasi_newton_method}, however, 
the self-scaling method with inexact line search for the step length 
tends to lead 
the relative inefficiency compared to the standard BFGS method. 
Following Example \ref{example:VBFGS-power-div} below, 
we prove that the self-scaling method with the scaling parameter
\eqref{eqn:popular-self-scaling-parameter} is not derived from 
the $V$-Bregman divergence. 

We present a practical way of computing the Hessian approximation
\eqref{eqn:update-formula-V-BFGS}. 
In Eq~\eqref{eqn:update-formula-V-BFGS}, the optimal solution $B_{k+1}$ appears in both 
sides, that is, we have only the implicit expression of $B_{k+1}$. 
The numerical computation is, however, efficiently performed as well as the standard BFGS
update. To compute the update formula $B_{k+1}$, first we compute $\det B_{k+1}$. 
The determinant of both sides of \eqref{eqn:update-formula-V-BFGS} leads to 
\begin{align}
 \det B_{k+1}=\frac{\det(B^{BFGS}[B_k;s_k,y_k])}{\nu(\det B_{k})^{n-1}}\cdot\nu(\det
 B_{k+1})^{n-1}. 
 \label{eqn:determinant-V-BFGS}
\end{align}
Hence, by solving the nonlinear equation 
\begin{align*}
 z=\frac{\det(B^{BFGS}[B_k;s_k,y_k])}{\nu(\det B_{k})^{n-1}}\cdot\nu(z)^{n-1},\qquad z>0
\end{align*}
we can find $\det B_{k+1}$. 
As shown in the proof of Theorem \ref{theorem:V-BFGS-form}, the function $z/\nu(z)^{n-1}$
is monotone increasing. Hence the Newton method is available to find the root of the above
equation efficiently. 
Once we obtain the value of $\det B_{k+1}$, we can compute the Hessian approximation $B_{k+1}$ by
substituting $\det B_{k+1}$ into Eq~\eqref{eqn:update-formula-V-BFGS}. 
Figure \ref{fig:V-BFGS-update} shows the update algorithm of the $V$-BFGS formula which
exploits the Cholesky decomposition of the approximate Hessian matrix. 
By maintaining the Cholesky decomposition, we can easily compute the the determinant and
the search direction. In the algorithm of Figure \ref{fig:V-BFGS-update}, we require the
Wolfe condition \cite[Section 3.1]{nocedal99:_numer_optim} for the step length
$\alpha_k$. 
As shown in Section \ref{sec:Convergence_Analysis}, the Wolfe condition is useful to
establish the convergence property of the optimization algorithm. 

In the same way as the proof of Theorem \ref{theorem:V-BFGS-form}, 
we obtain the $V$-DFP update formula defined from \eqref{eqn:V-DFP-hessian-update-prob}
such that 
\begin{align}
 B_{k+1}=\frac{\nu((\det{B_k})^{-1})}{\nu((\det{B_{k+1}})^{-1})}B^{DFP}[B_k;s_k,y_k]
 +\bigg(1-\frac{\nu((\det{B_k})^{-1})}{\nu((\det{B_{k+1}})^{-1})}\bigg)
 \frac{y_ky_k^\top}{s_k^\top y_k}. 
 \label{eqn:update-formula-V-DFP}
\end{align}
It is straightforward to unify the $V$-BFGS method and the $V$-DFP method 
in the same way as the standard Broyden family
\cite{broyden67:_quasi_newton_method_and_their}. 
Let $B_{V_1,k+1}^{\mathrm{BFGS}}$ be the Hessian approximation given by the $V$-BFGS
update formula with the potential $V=V_1$, and 
$B_{V_2,k+1}^{\mathrm{DFP}}$ be the Hessian approximation given by the $V$-DFP update 
formula with the potential $V=V_2$. 
Then the update formula of the $(V_1,V_2)$-Broyden family is defined by 
\begin{align}
 B_{k+1}~=~\vartheta\, B_{\mathrm{BFGS},k+1}^{(V_1)} +
 (1-\vartheta)\,B_{\mathrm{DFP},k+1}^{(V_2)},
 \label{eqn:extended-Broyden-family}
\end{align}
for $\vartheta\in[0,1]$. 
The $(V_1,V_2)$-Broyden family is obtained by a convex-full of $B^{BFGS}[B_k;s_k,y_k]$, 
$B^{DFP}[B_k;s_k,y_k]$ and $y_ky_k^\top/s_k^\top y_k$. The standard Broyden family is
recovered by setting $V_1(z)=V_2(z)=-\log z$. 

\begin{figure}[p]
 \label{fig:V-BFGS-update}
 \centering 
 \fbox{
 \begin{minipage}{0.9\linewidth}
\begin{description}
 \item[$V$-BFGS update:] 
 \item[Initialization:] 
            The function $\nu(z)$ denotes $-V'(z)z$. 
            Let $B_0\in\mathrm{PD}(n)$ be a matrix which is an initial
            approximation of the Hessian matrix, and $L_0L_0^\top=B_0$ be 
            the Cholesky decomposition of $B_0$. Let $x_0\in\Real^n$ be an initial point,
            and set $k=0$. 
 \item[Repeat:] If stopping criterion is satisfied, go to Output. 
            \begin{enumerate}
             \item Let $x_{k+1}=x_k-\alpha_k B_k^{-1}\nabla f(x_k)$, 
		   where $\alpha_k\geq 0$ is a step length satisfying the Wolfe condition
		   \cite[Section 3.1]{nocedal99:_numer_optim}. 
		   The Cholesky decomposition $B_k=L_kL_k^\top$ is available to compute
                   $B_k^{-1}\nabla f(x_k)$. 
	     \item Set $s_k=x_{k+1}-x_k$ and $y_k=\nabla f(x_{k+1})-\nabla f(x_k)$. 
             \item Update $L_k$ to $\bar{L}$ which is the Cholesky decomposition 
		   of $B^{BFGS}[B_k;s_k,y_k]$, that is, 
                   \begin{align*}
                    \bar{L}\bar{L}^\top
		    =B^{BFGS}[B_k;s_k,y_k]=B^{BFGS}[L_kL_k^\top;s_k,y_k]. 
                   \end{align*}
                   The Cholesky decomposition with rank-one update is available. 
             \item Compute 
		   \begin{align*}
		    C=\frac{(\det{\bar{L}})^2}{\nu((\det{L_k})^2)^{n-1}}
		   \end{align*}
		   and find the root of the equation 
                   \begin{align*}
                    C\cdot\nu(z)^{n-1} = z,\qquad z>0. 
                   \end{align*}
                   Let the solution be $z^*$. 
             \item Compute the Cholesky decomposition $L_{k+1}$ such that 
                   \begin{align*}
		    L_{k+1}L_{k+1}^\top = 
		    \frac{\nu(z^{*})}{\nu((\det{L_k})^2)}
		    \bar{L}\bar{L}^\top 
                    +\bigg(1-\frac{\nu(z^{*})}{\nu((\det{L_k})^2)}\bigg) 
		    \frac{y_ky_k^\top}{s_k^\top y_k}. 
                   \end{align*}
             \item $k\leftarrow k+1$. 
            \end{enumerate}
 \item[Output:]  Local optimal solution $x_{k}$. 
\end{description}
 \end{minipage}}
 \caption{Pseudo code of $V$-BFGS method. 
 The Cholesky decomposition with rank-one update is useful in the algorithm. }
 \label{fig:Adaboost.alg}
\end{figure}

\begin{example}
 \label{example:VBFGS-power-div}
 We show the $V$-BFGS formula derived from the power potential. 
 Let $V(z)$ be the power potential $V(z)=(1-z^\gamma)/\gamma$ with $\gamma<1/n$. 
 As shown in Example \ref{example:power-div}, we have $\nu(z)=z^\gamma$. 
 Due to the equality 
 \begin{align*}
  \det(B^{BFGS}[B_k;s_k,y_k])=\det(B_k)\frac{s_k^\top y_k}{s_k^\top B_ks_k}
 \end{align*}
 and Eq.~\eqref{eqn:determinant-V-BFGS}, for the power potential we have 
 \begin{align*}
  \frac{\nu(\det{B_{k+1}})}{\nu(\det{B_k})}=
  \left(\frac{s_k^\top y_k}{s_k^\top B_ks_k}\right)^\rho,\qquad
  \rho=\frac{\gamma}{1-(n-1)\gamma}. 
 \end{align*}
 Then the $V$-BFGS update formula is given as 
\begin{align*}
 B_{k+1} =
\left(\frac{s_k^\top y_k}{s_k^\top B_ks_k}\right)^\rho
 B^{BFGS}[B_k;s_k,y_k]
 +\bigg(1-\left(\frac{s_k^\top y_k}{s_k^\top B_ks_k}\right)^\rho
 \bigg)
 \frac{y_ky_k^\top}{s_k^\top y_k}. 
\end{align*}
 For $\gamma$ such that $\gamma<1/n$, we have $-1/(n-1)<\rho<1$. 
 Remember that the standard self-scaling update formula corresponds to the above update
 with $\rho=1$. Therefore, the standard self-scaling update formula is not derived from
 the power potential. Indeed, the power potential with $\rho=1$ or equivalently
 $\gamma=1/n$ is a convex function but not a strictly convex function. 
\end{example}
In terms of the self-scaling update formula, we show the following proposition. 
\begin{proposition}
 There does not exist the potential function such that 
 in Eq.~\eqref{eqn:update-formula-V-BFGS} the equality
 \begin{align}
  \frac{\nu(\det{B_{k+1}})}{\nu(\det{B_{k}})}=\frac{s_k^\top y_k}{s_k^\top B_ks_k}
  \label{eqn:potential-condition-popular-self-scaling}
 \end{align}
 holds for any $B_k\in\PD(n)$ and any $s_k, y_k\in\Real^n$ satisfying $s_k^\top y_k>0$. 
\end{proposition}
\begin{proof}
 We have two equalities, 
 \begin{align*}
  \det(B^{BFGS}[B_k;s_k,y_k])&=\det(B_k)\frac{s_k^\top y_k}{s_k^\top B_ks_k},\\
  \det{B_{k+1}}&=\frac{\det(B^{BFGS}[B_k;s_k,y_k])}{\nu(\det{B_k})^{n-1}}\nu(\det{B_{k+1}})^{n-1}. 
 \end{align*}
 Hence, we have 
 \begin{align*}
  \left(\frac{\nu(\det{B_{k+1}})}{\nu(\det{B_k})}\right)^{n-1}=
  \frac{\det{B_{k+1}}}{\det{B_k}}\cdot\frac{s_k^\top B_ks_k}{s_k^\top y_k}
 \end{align*}
 Suppose that there exists a potential function satisfying
 \eqref{eqn:potential-condition-popular-self-scaling}. Then we have 
\begin{align*}
 \left(\frac{s_k^\top y_k}{s_k^\top B_ks_k}\right)^{n-1}=
  \frac{\det{B_{k+1}}}{\det{B_k}}\cdot\frac{s_k^\top B_ks_k}{s_k^\top y_k}, 
\end{align*}
and hence the equality 
\begin{align*}
 \det{B}_{k+1}=\det({B}_{k})\cdot \left(\frac{s_k^\top y_k}{s_k^\top B_ks_k}\right)^{n}
\end{align*}
holds. Substituting the above formula into
 \eqref{eqn:potential-condition-popular-self-scaling}, 
 we have 
\begin{align*}
 \nu\left(\det({B}_{k}) \left(\frac{s_k^\top y_k}{s_k^\top B_ks_k}\right)^{n}\right)=
 \nu(\det{B}_k)\frac{s_k^\top y_k}{s_k^\top B_ks_k}. 
\end{align*}
 Let $B_k$ be a positive definite matrix such that $\det{B}_k=1$, and 
 $z$ be $z=\left(\frac{s_k^\top y_k}{s_k^\top B_ks_k}\right)^{n}$. 
 Then we have $\nu(z)=\nu(1)z^{1/n}$ for $z>0$. 
 The corresponding $\beta_V$ is given as $\beta_V(z)=1/n$, and 
 this does not satisfy the definition of the potential function. 
\end{proof}

\section{Convergence Analysis}
\label{sec:Convergence_Analysis}
We consider the convergence property of the $V$-BFGS method. 
Some standard assumptions about the objective function $f$ are stated below. See Section
6.4 of \cite{nocedal99:_numer_optim} for details. 
\begin{assumption}
 \label{assumption:convergence}
 \begin{enumerate}
  \item The objective function $f$ is twice continuously differentiable. 
  \item Let $\nabla^2 f(x)$ be the Hessian matrix of $f$ at $x$. 
	For the starting point $x_0$, the level set 
	$\mathcal{L}=\{x\in\Real^n~|~f(x)\leq f(x_0)\}$ is convex, and there 
	exist positive constants $m$ and $M$ such that
	\begin{align}
	 \label{eqn:assumption:convergence-eigen}
	 m\|z\|^2 \leq z^\top \nabla^2 f(x) z\leq M\|z\|^2
	\end{align}
	holds for all $z\in\Real^n$ and $x\in\mathcal{L}$. 
 \end{enumerate}
\end{assumption}
The following theorem implies that the sequence $\{x_k\}$ generated by the $V$-BFGS update
formula converges to the local minimizer of $f$ if the function $\nu_V$ of a potential $V$
satisfies the bounding condition. 
\begin{theorem}
 \label{theorem:convergence}
 Let $B_0\in\mathrm{PD}(n)$ be an initial matrix and $x_0\in\Real^n$ be a starting point
 which meets Assumption \ref{assumption:convergence}. 
 Suppose that there exist positive constants $L_1, L_2>0$ such that 
 $L_1\leq \nu\leq L_2$. 
 Then the sequence $\{x_k\}$ generated by the $V$-BFGS update converges to the minimizer
 $x^*$ of $f$. 
\end{theorem}
\begin{lemma}[Eq. 6.12 in \cite{nocedal99:_numer_optim}]
 \label{lemma:averaged_Hessian}
 Let $\bar{G}$ be the averaged Hessian
 \[
  \bar{G}=\int_0^1 \nabla^2 f(x_k+\tau s)d\tau,\quad s=x_{k+1}-x_k\in\Real^n, 
 \]
 then the property $y=\bar{G} s$ follows from Taylor's
 theorem, where $y=\nabla f(x_{k+1})-\nabla f(x_k)$. 
\end{lemma}
Using Lemma \ref{lemma:averaged_Hessian}, we prove Theorem \ref{theorem:convergence} in a
manner similar to Section 8.4 in \cite{nocedal99:_numer_optim}. 
\begin{proof}[Proof of Theorem \ref{theorem:convergence}]
Let $B_k, k=0,1,2,\ldots$ be the sequence of approximate Hessian matrices generated by the 
$V$-BFGS update formula. 
We define $\bar{B}_{k+1}$ and $\bar{B}_k$ by 
$\bar{B}_{k+1}=\frac{1}{\nu(\det B_{k+1})}B_{k+1}$ and 
$\bar{B}_k=\frac{1}{\nu(\det B_k)}B_k$, respectively. 
Then the update formula shown in Theorem \ref{theorem:V-BFGS-form} is represented as 
\begin{align}
 \label{eqn:proof-v-bggs-update-formula}
 \bar{B}_{k+1}&=\bar{B}_k-\frac{\bar{B}_ks_ks_k^\top \bar{B}_k}{s_k^\top \bar{B}_k s_k}
 + \frac{1}{\nu(\det{B_{k+1}})}\frac{y_ky_k^\top}{s_k^\top y_k}. 
\end{align}
 We compute 
 \begin{align*}
  \psi(\bar{B}_{k+1})=\tr{\bar{B}_{k+1}}-\log\det\bar{B}_{k+1}. 
 \end{align*}
 The inequality \eqref{eqn:assumption:convergence-eigen} yields 
\begin{align}
 \label{eqn:proof-eigen-inequ-1}
 \frac{s_k^\top y_k}{\|s_k\|^2}
 &=\frac{s_k^\top \bar{G} s_k}{\|s_k\|^2}\geq m,\\
 \label{eqn:proof-eigen-inequ-2}
 \frac{\|y_k\|^2}{s_k^\top y_k}
 &=\frac{s_k^\top \bar{G}^2 s_k}{s_k^\top\bar{G}s_k}\leq M. 
\end{align}
We now define
\begin{align*}
 \cos\theta_k=\frac{s_k^\top \bar{B}_k s_k}{\|s_k\| \|\bar{B}_k s_k\|},\qquad
 q_k=\frac{s_k^\top \bar{B}_k s_k}{\|s_k\|^2}. 
\end{align*}
Then the trace of $\bar{B}_{k+1}$ is bounded above. Indeed, the inequality 
\begin{align*}
 \tr{\bar{B}_{k+1}}
 &=
 \tr{\bar{B}_k}-\frac{\|\bar{B}_k s_k\|^2}{s_k^\top \bar{B}_k s_k}+
 \frac{\|y_k\|^2}{\nu(\det{B_{k+1}})s_k^\top y_k}
 \leq 
 \tr{\bar{B}_k}-\frac{q_k}{\cos^2\theta_k}+\frac{M}{\nu(\det{B_{k+1}})}, 
\end{align*}
holds, where \eqref{eqn:proof-eigen-inequ-2} is used. 
Using the formula 
$\det(I+xy^\top+uv^\top)=(1+x^\top y)(1+u^\top v^\top)-(x^\top v)(y^\top u)$ for 
$\bar{B}_{k+1}$, we obtain a lower bound of the determinant $\det(\bar{B}_{k+1})$ such that 
\begin{align*}
 \det(\bar{B}_{k+1}) 
 &= 
 \det(\bar{B}_k)
 \frac{1}{\nu(\det{B_{k+1}})}\frac{\|s_k\|^2}{s_k^\top \bar{B}_k s_k}
 \frac{s_k^\top y_k}{\|s_k\|^2} 
 \geq \det(\bar{B}_k) \frac{m}{q_k\nu(\det{B_{k+1}})}. 
\end{align*}
These inequalities present an upper bound of $\psi(\bar{B}_{k+1})$, 
\begin{align*}
 \psi(\bar{B}_{k+1})
 &\leq
 \psi(\bar{B}_k)
 +\bigg(\frac{M}{\nu(\det{B_{k+1}})}-\log\frac{m}{\nu(\det{B_{k+1}})}-1\bigg)\\
 &\phantom{\leq}
 +\bigg(1-\frac{q_k}{\cos^2\theta_k}+\log\frac{q_k}{\cos^2\theta_k}\bigg)+\log\cos^2\theta_k\\
 &\leq 
 \psi(\bar{B}_k)
 +\bigg(\frac{M}{L_1}-\log\frac{m}{L_2}-1\bigg)
 +\log\cos^2\theta_k. 
\end{align*}
The second inequality is derived from
\begin{align*}
1-\frac{q_k}{\cos^2\theta_k}+\log\frac{q_k}{\cos^2\theta_k}\leq 0. 
\end{align*}
As the result we obtain
\begin{align*}
 0<\psi(\bar{B}_{k+1})\leq  \psi(\bar{B}_0)+c(k+1)+\sum_{j=1}^k\log\cos^2\theta_j, 
\end{align*}
where $c$ is a positive constant such that $c>\frac{M}{L_1}-\log\frac{m}{L_2}-1$. 
Let us then proceed by contradiction and assume that $\cos\theta_j\rightarrow 0$. Then
 there exists $k_1>0$ such that for all $j>k_1$, we have 
 \begin{align*}
 \log \cos^2\theta_j<-2c. 
 \end{align*}
Thus the following inequality holds for all $k>k_1$:
 \begin{align*}
  0
  &<\psi(\bar{B}_0)+c(k+1)+\sum_{j=1}^{k_1}\log\cos^2\theta_j+(k-k_1)(-2c) \\
  &=\psi(\bar{B}_0)+\sum_{j=1}^{k_1}\log\cos^2\theta_j+c(2k_1+1)-2ck. 
 \end{align*}
The right-hand-side is negative for large $k$, giving a contradiction. Therefore there
exists a subsequence satisfying $\cos\theta_{j_k}\geq \delta>0$. 
By Zoutendijk's result\footnote{Under some condition, 
$\sum_{j\geq0}\cos^2\theta_j\|\nabla f(x_j)\|^2<\infty$ holds. See Theorem 3.2 in
\cite{nocedal99:_numer_optim}}
with the Wolfe condition, 
this limit implies that $\liminf_{k\rightarrow\infty}\|\nabla f(x_k)\|=0$. 
The convexity of $f$ on $\mathcal{L}$ guarantees that $x_k$ converges to the local optimal
solution. 
\end{proof}
The potential defined in Example \ref{example:bounded-div} meets the condition of Theorem
\ref{theorem:convergence}, while the power potential $V(z)=(1-z^\gamma)/\gamma$ with
$\nu(z)=z^\gamma$ does not satisfy the condition.

\section{Robustness against Inexact Line Search}
\label{sec:Robustness}
The robustness against numerical errors such as the round-off error is an important
feature in numerical computation. 
In this section we study the robustness of quasi-Newton update against numerical
errors involved in the line search. 
Mainly there are two types of quasi-Newton updates: 
one is the update formula for approximate Hessian matrix;  and the other is the update for approximate
{\em inverse} Hessian matrix. In the approximate inverse Hessian update, the matrix 
$H_k=B_k^{-1}$ is directly update to $H_{k+1}=B_{k+1}^{-1}$ under the secant condition 
$H_{k+1}y_k=s_k$. 
We study four kinds of update formulae, that is, $V$-BFGS/$V$-DFP method for 
the Hessian approximation/the inverse Hessian approximation. 

Let us consider the Hessian approximation formula. 
Under the exact line search, the matrix $B_k$ is updated to $B_{k+1}$ which is the 
minimum solution of $D_V(B,B_k)$ or $D_V(B^{-1},B_k^{-1})$ subject to  $Bs_k=y_k$. Let 
\begin{align*}
x_{k+1}=x_k-\alpha_k B_k^{-1}\nabla f(x_k)=x_k+s_k
\end{align*}
be the point computed by the exact line search. When the line search is inexact, 
the step length $\alpha_k$ will be slightly perturbed and then $s_k$ will be changed to 
$(1+\varepsilon)s_k$ where $\varepsilon$ is an infinitesimal. 
The vector $y_k$ will also change to $\widetilde{y}_k$ defined by
\begin{align*}
 \widetilde{y}_k=\nabla f(x_k+(1+\varepsilon)s_k)-\nabla f(x_k)=y_k+\varepsilon\nabla^2 
 f(x_{k+1})s_k+O(\varepsilon^2). 
\end{align*}
Then the constraint for the Hessian update becomes $(1+\varepsilon)Bs_k=\widetilde{y}_k$. 

We study the relation between the perturbation of $s_k$ and the Hessian approximation
$B_{k+1}$ or the inverse Hessian approximation $H_{k+1}$. 
Based on the above argument, we consider the optimization problem defined by 
\begin{align}
 \label{eqn:perturbed-V-BFGS}
 \text{($V$-BFGS-B) }&\quad
 \min_{B\in \mathrm{PD}(n)} D_V(B,B_k)
 \quad \subto\ \ (1+\varepsilon)Bs=y+\varepsilon
 \bar{y},\\
 \label{eqn:perturbed-V-DFP}
 \text{($V$-DFP-B) }&\quad
 \min_{B\in \mathrm{PD}(n)} D_V(B^{-1},B_k^{-1})
 \quad \subto\ \ (1+\varepsilon)Bs=y+\varepsilon
 \bar{y}
\end{align}
for a fixed matrix $B_k\in\mathrm{PD}(n)$ and fixed vectors $s, y, \bar{y}\in\Real^n$,
where the subscript $k$ for the vectors is dropped for simplicity. 
In the same way, the update formula for the inverse Hessian under the inexact line search
is defined as the optimal solution of the following problem, 
\begin{align}
 \label{eqn:perturbed-V-BFGS-H}
 \text{($V$-BFGS-H) }&\quad
 \min_{H\in \mathrm{PD}(n)} D_V(H^{-1},H_k^{-1})\quad 
 \subto\ \ H(y+\varepsilon\bar{y})=(1+\varepsilon)s, \\
 \label{eqn:perturbed-V-DFP-H}
 \text{($V$-DFP-H) }&\quad
 \min_{H\in \mathrm{PD}(n)} D_V(H,H_k)\quad 
 \subto\ \ H(y+\varepsilon\bar{y})=(1+\varepsilon)s, 
\end{align}
for fixed $H_k\in\mathrm{PD}(n),\, s, y, \bar{y}\in\Real^n$. 
The update formula given by $V$-BFGS-H/$V$-DFP-H directly provides the inverse matrix of 
$B_{k+1}$ computed by $V$-BFGS-B/$V$-DFP-B, respectively. 
Theorem \ref{theorem:V-BFGS-form} guarantees that there exists the unique optimal solution 
as long as $s^\top(y+\varepsilon\bar{y})>0$ holds. 
Though Theorem \ref{theorem:V-BFGS-form} deals with only $V$-BFGS-B formula, 
we can prove the existence and the uniqueness of optimal solution for the other problems
in the same manner. 

In order to study the robustness of update formulae, 
we borrow the concepts such that the influence function or the gross error sensitivity
from the study of robust statistics \cite{Hampel_etal86}. 
Below the $V$-BFGS-B update formula is considered as an example. 
Let $B(\varepsilon)$ be the optimal solution of $V$-BFGS-B in 
\eqref{eqn:perturbed-V-BFGS}. Then the {\em influence function} of $B(\varepsilon)$ is 
defined as the derivative of $B(\varepsilon)$ at $\varepsilon=0$, that is, 
\begin{align*}
\dot{B}(0)=\lim_{\varepsilon\rightarrow0}
 \frac{B(\varepsilon)-B(0)}{\varepsilon}.
\end{align*}
Later we prove the differentiability of $B(\varepsilon)$. 
From the definition of the influence function, 
the optimal solution $B(\varepsilon)$ is asymptotically equal to
$B(0)+\varepsilon\dot{B}(0)$. This implies that 
the inexact line search has a large impact on the computation of 
Hessian approximation, when the norm of $\dot{B}(0)$ is large. 
In the sense of the influence function, the preferable potential is the function $V$ 
which provides the influence function $\dot{B}(0)$ with a small norm. 

For fixed vectors $s$ and $y$ such that $s^\top y>0$, the influence function $\dot{B}(0)$ 
depends on the matrix $B_k$ and the vector $\bar{y}$. 
We consider the worst-case evaluation of the influence function in terms of $B_k$ and
$\bar{y}$. The {\em gross error sensitivity} is defined as the largest norm of the influence
function, that is, 
\begin{align*}
 \text{gross error sensitivity}=\sup
 \big\{\|\dot{B}(0)\|_F~|~
 B_k\in\mathcal{B}\subset\mathrm{PD}(n),\,\bar{y}\in
 {\mathcal{Y}}\subset\Real^n\big\}, 
\end{align*}
where $\mathcal{B}\subset\mathrm{PD}(n)$ and $\mathcal{Y}\subset\Real^n$ are 
appropriate subsets. In many case, the gross error sensitivity becomes infinity if
$\mathcal{B}$ or ${\mathcal{Y}}$ is unbounded. 
Our concern is to find the potential function $V$ which leads finite gross error
sensitivity under some reasonable setup. 

The influence function and the gross error sensitivity have been studied in robust
statistics \cite{Hampel_etal86}. 
We use these statistical techniques to analyze the stability of numerical computation. 
In the literature of statistics, the ``statistical model'' $\{B\in\mathrm{PD}(n)~|~Bs_k=y_k\}$ 
or $\{H\in\mathrm{PD}(n)~|~Hy_k=s_k\}$ is fixed, 
and the ``observed data'' $B_k$ or $H_k$ is contaminated such that 
$B_k+\varepsilon \dot{B}(0)+O(\varepsilon^2)$, 
while in the present analysis, the matrix $B_k=H_k^{-1}$ is fixed and 
the model corresponding to the secant condition is perturbed. 

The potential function minimizing the gross error sensitivity will be preferable for
robust computation. 
Below we prove that the standard BFGS update for the Hessian approximation 
is the more robust than the other update formulae. This result meets the empirical 
observations \cite{conn88:_testin_class_of_algor_for,nocedal99:_numer_optim}. 
Moreover, only the standard BFGS update for the Hessian approximation
has finite gross error sensitivity. Theoretical results are summarized in Table 
\ref{tbl:gross_error_sesitivity}. 

\begin{table*}[tb]
\caption{Gross error sensitivity of $V$-BFGS formula and $V$-DFP formula for the Hessian 
 approximation and the inverse Hessian approximation. Only the standard BFGS for the
 Hessian approximation has finite gross error sensitivity.} 
 \label{tbl:gross_error_sesitivity}
 \centering\vspace*{2mm}
\begin{tabular}{c|c|c} 
                    & $V$-BFGS              & {\qquad \,\, $V$-DFP\qquad\,\,  }   \\\hline
  Hessian approx.       & finite only for BFGS  & $\infty$  \\\hline
 inverse Hessian approx.& $\infty$              & $\infty$  \\ 
\end{tabular}
\end{table*}

In the following, the gross error sensitivity with $\mathcal{B}=\mathrm{PD}(n)$ and a
bounded subset ${\mathcal{Y}}$ is considered. 
Note that the boundedness of $\mathcal{Y}$ follows the assumption that $\|\nabla^2f\|_F$
is bounded above over $\Real^n$. 
First, we note that the influence function and the gross error sensitivity make sense for
minimization of non-quadratic functions. 
\begin{lemma}
 \label{lemma:quadratic_function_grosserror}
 Suppose that the objective function $f(x)$ is a convex quadratic function. 
 Then, the influence function and the gross error sensitivity are equal to zero. 
\end{lemma}
Lemma \ref{lemma:quadratic_function_grosserror} is clear, since for the quadratic
objective function the secant condition $Bs=y$ is changed to
$B(1+\varepsilon)s=(1+\epsilon)y$ under the inexact line search. 
That is, the secant condition is kept unchanged, and thus $B(\varepsilon)=B(0)$ holds. 

We prove that generally the influence function is well-defined. 
\begin{theorem}
 \label{theorem:existence-differentiability}
 Suppose that $s^\top y>0$ holds for vectors $s$ and $y$ in 
 the problems \eqref{eqn:perturbed-V-BFGS}, \eqref{eqn:perturbed-V-DFP}, 
 \eqref{eqn:perturbed-V-BFGS-H} and \eqref{eqn:perturbed-V-DFP-H}. 
 Then, for small $\varepsilon$, the optimal solutions 
 of $V$-BFGS-B, $V$-DFP-B, $V$-BFGS-H and $V$-DFP-H are all uniquely determined. 
 The optimal solutions are second-order continuously differentiable with respect to
 $\varepsilon$ in the vicinity of $\varepsilon=0$.  
\end{theorem}
Proof is deferred to Appendix \ref{appendix:proof_theorem_existence_optsol}. 

The gross error sensitivity of each update formula is computed in the following theorems. 
Proofs are deferred to Appendix \ref{appendix:robustness-sensitivity}. 
\begin{theorem}[gross error sensitivity of $V$-BFGS-B]
 \label{theorem:BFGS-B-robustness}
 Suppose $n\geq 3$. 
 Let $s$ and $y$ be fixed vectors such that $s^\top y>0$
 and ${\mathcal{Y}}$ be a bounded subset in $\Real^n$. 
 For small $\varepsilon$, let $B(\varepsilon)$ be the optimal solution of $V$-BFGS-B in 
 \eqref{eqn:perturbed-V-BFGS}. 
 Then, the optimal potential function of the problem 
 \begin{align}
  \label{eqn:min-max-general-V}
  \min_{V}\max_{B_k,\,\bar{y}}\|\dot{B}(0)\|_F\quad 
  \subto\ B_k\in\mathrm{PD}(n),\ \ \bar{y}\in{\mathcal{Y}}
 \end{align}
 is given as $V(z)=-\log(z)$ up to a constant factor. 
 In the above min-max problem, the function $V$ is sought from among all potentials. 
\end{theorem}

\begin{theorem}[gross error sensitivity of $V$-DFP-B]
 \label{theorem:DFP-B-robustness}
 Suppose $n\geq 3$. 
 Let $s$ and $y$ be fixed vectors such that $s^\top y>0$
 and ${\mathcal{Y}}$ be a bounded subset in $\Real^n$. 
 Suppose that there exists an open subset included in ${\mathcal{Y}}$. 
 Let $B(\varepsilon)$ be the optimal solution of $V$-DFP-B in
 \eqref{eqn:perturbed-V-DFP}. 
 Then for any potential $V$, the equality
 \begin{align*}
  \sup\{\|\dot{B}(0)\|_F~|~ B_k\in\mathrm{PD}(n),\, 
  \bar{y}\in{\mathcal{Y}}\}=\infty
 \end{align*}
 holds. 
\end{theorem}

\begin{theorem}[gross error sensitivity of $V$-BFGS-H]
 \label{theorem:BFGS-H-robustness}
 Suppose $n\geq 4$. 
 Let $s$ and $y$ be fixed vectors such that $s^\top y>0$
 and ${\mathcal{Y}}$ be a bounded subset in $\Real^n$. 
 Suppose that there exists an open subset included in ${\mathcal{Y}}$. 
 Let $H(\varepsilon)$ be the optimal solution of $V$-BFGS-H in
 \eqref{eqn:perturbed-V-BFGS-H}. 
 Then, for any potential $V$, the equality
 \begin{align*}
  \sup\{\|\dot{H}(0)\|_F~|~ H_k\in\mathrm{PD}(n),\, 
  \bar{y}\in{\mathcal{Y}}\}=\infty
 \end{align*}
 holds. 
\end{theorem}

\begin{theorem}[gross error sensitivity of $V$-DFP-H]
 \label{theorem:DFP-H-robustness}
 Suppose $n\geq 3$. 
 Let $s$ and $y$ be fixed vectors such that $s^\top y>0$
 and ${\mathcal{Y}}$ be a bounded subset in $\Real^n$. 
 Let $H(\varepsilon)$ be the optimal solution of $V$-DFP-H in
 \eqref{eqn:perturbed-V-DFP-H}. 
 Then, for any potential $V$, the equality
 \begin{align*}
  \sup\{\|\dot{H}(0)\|_F~|~ H_k\in\mathrm{PD}(n),\, 
  \bar{y}\in{\mathcal{Y}}\}=\infty
 \end{align*}
 holds. 
\end{theorem}

It is well-known that there is the dual relation between the BFGS formula and the DFP formula. 
Indeed, the $V$-DFP update for the inverse Hessian approximation is derived from the
$V$-BFGS update formula for the Hessian approximation by replacing $B_k,s_k,y_k$ with 
$H_k,y_k,s_k$. For the robustness against inexact line search, however, the dual relation
is violated as shown in Table \ref{tbl:gross_error_sesitivity}. 
In this problem, we focus on the perturbation of the vector $s_k$ rather than that of $y_k$. 
This is the reason why the dual relation is violated. 
Powell has shown a critical difference between BFGS and DFP for quadratic convex objective
functions \cite{powell86:_how_bad_are_bfgs_and} by considering the behaviour of
eigenvalues of approximate Hessian matrix. In the present paper, we exploited the gross
error sensitivity which is meaningful for non-quadratic objective functions as shown 
in Lemma \ref{lemma:quadratic_function_grosserror}. Our approach also provides a critical
difference between BFGS and DFP methods. 

In Section \ref{sec:extension_quasi-Newton_updates}, 
we introduced the $(V_1,V_2)$-Broyden family defined by
\eqref{eqn:extended-Broyden-family}.  
It is straightforward to prove that only the standard BFGS has finite gross error 
sensitivity among the $(V_1,V_2)$-Broyden family with a fixed mixing parameter
$\vartheta\in[0,1]$.

\section{Numerical Studies}
\label{sec:Numerical_Studies}
We demonstrate numerical experiments on robustness of quasi-Newton update formulae such as
$V$-BFGS-B, $V$-DFP-B, $V$-BFGS-H, and $V$-DFP-H proposed in Section
\ref{sec:Robustness}. Especially, the update formula derived from power potential 
in Example \ref{example:power-div} is examined. 

In the first numerical study, we consider numerical stability of update formulae. 
Let $B(\varepsilon)$ be the optimal 
solution of $V$-BFGS-B \eqref{eqn:perturbed-V-BFGS} or $V$-DFP-B
\eqref{eqn:perturbed-V-DFP}, and $H(\varepsilon)$ be the optimal solution of $V$-BFGS-H
\eqref{eqn:perturbed-V-BFGS-H} or $V$-DFP-H \eqref{eqn:perturbed-V-DFP-H}. 
For each update formula, we numerically compute the approximate influence function  
$\|(B(\varepsilon)-B(0))/\varepsilon\|_F$ and $\|(H(\varepsilon)-H(0))/\varepsilon\|_F$ 
with small $\varepsilon$, where the power potential $V(z)=(1-z^\gamma)/z$ is used to
derive the approximate Hessian matrix. 
Remember that $V$-BFGS and $V$-DFP are respectively reduced to the standard BFGS and DFP
when $\gamma$ is equal to zero. 

In what follows, we show the setup of numerical studies. 
Let $\mathrm{diag}(a_1,\ldots,a_n)$ be the $n$ by $n$ diagonal matrix with diagonal
elements $a_1,\ldots,a_n$. 
For $V$-BFGS-B and $V$-DFP-B, 
the matrix $B_k$ is set to one of the following three matrices:  
\begin{align*}
 B_k=\mathrm{diag}(1,\ldots,n)/(n!)^{1/n},\quad 
 B_k=\mathrm{diag}(1,\ldots,n),\quad \text{or}\quad 
 B_k=I+n^3\cdot pp^\top, 
\end{align*}
where in the last one $I$ is the identity matrix and $p$ is a column unit vector defined below. 
The dimension of the matrix $B_k$ is set to $n=10, 100, 500$ or $1000$. 
The first matrix $\mathrm{diag}(1,\ldots,n)/(n!)^{1/n}$ has the determinant one, and the
other two matrices have a large determinant. 
Below we show the procedure for generating the vectors $s$ and $y$ and the contaminated
vectors $(1+\varepsilon)s$ and $y+\varepsilon\bar{y}$ for 
$V$-BFGS-B and $V$-DFP-B. 
In the numerical studies for $V$-BFGS-H and $V$-DFP-H, 
the matrix $B_k$ is replaced with the approximate inverse Hessian $H_k$. 
\begin{enumerate}
 \item In the case that $B_k$ is $\mathrm{diag}(1,\ldots,n)/(n!)^{1/n}$ or
       $\mathrm{diag}(1,\ldots,n)$, 
       the vectors $s$ and $y$ are both generated according to the multivariate normal
       distribution with mean zero and variance-covariance matrix $10\times I$. 
       If the inner product $s^\top y$ is non-positive, the sign of $y$ is flipped. 
       The intensity of noise involved in the line search
       is determined by $\varepsilon$, which is 
       generated according to the uniform distribution on the interval 
       $[-0.2,0.2]$. Then, the vector $\bar{y}$ is also generated according to the multivariate 
       standard normal distribution. 
       If the inequality $(1+\varepsilon)s^\top(y+\varepsilon\bar{y})>0$ does not hold, 
       again $\varepsilon$ and $\bar{y}$ are generated until the vectors enjoy the
       positivity condition. 
 \item In the case that $B_k$ is supposed to have the expression $I+n^3\cdot pp^\top$, 
       first the vector $s$ is generated according to the multivariate normal distribution 
       with mean zero and variance-covariance matrix $10\times I$, and 
       $y$ is defined such that $y=s$. 
       The vector $p$ is a unit vector which is orthogonal to $y$, that is, $p$ 
       is a vector satisfying $p^\top y=0$ and $\|p\|=1$, and 
       let $B_k$ be $B_k=I+n^3\cdot pp^\top$. 
       Then the vector $\bar{y}$ is defined as $\bar{y}=p$. 
       The construction of these vectors is used in the proof of 
       Theorem \ref{theorem:BFGS-H-robustness} and 
       Theorem \ref{theorem:DFP-H-robustness}. 
\end{enumerate}

Hessian or inverse Hessian update formula is applied to $B_k$ or $H_k$ with the randomly
generated secant condition. 
The updated matrix $B(0)$ and $B(\varepsilon)$ are respectively computed under 
the constraint  $Bs=y$ and $B(1+\varepsilon)s=y+\varepsilon\bar{y}$ by using 
$V$-BFGS-B and $V$-DFP-B update formula. 
In the same way, 
$V$-BFGS-H and $V$-DFP-H are respectively applied to compute $H(0)$ with the constraint
$Hy=s$ and $H(\varepsilon)$ with the perturbed secant condition
$H(y+\varepsilon\bar{y})=(1+\varepsilon)s$. 
The influence function of each update formula is approximated by
$\|(B(\varepsilon)-B(0))/\varepsilon\|_F$ or 
$\|(H(\varepsilon)-H(0))/\varepsilon\|_F$. 

Table \ref{tbl:gross_error_sensitivity} shows the average of the approximate influence
function over $20$ runs for each setup. 
When $B_k$ or $H_k$ is equal to the diagonal matrix $\mathrm{diag}(1,\ldots,n)/(n!)^{1/n}$, 
we see that the power $\gamma$ of the power potential 
does not significantly affect the influence function in both $V$-BFGS and $V$-DFP. 
For the other setups, overall the BFGS method for Hessian matrix, i.e. $V$-BFGS-B with
$\gamma=0$, has smaller influence function than the other update formulae. 
The $V$-DFP-H for inverse Hessian update also has relatively small influence function 
when $H_k$ is proportional to $\mathrm{diag}(1,\ldots,n)$. 
For $H_k=I+n^3pp^\top$, however, we find that $V$-DFP-H is sensitive against noise
involved in the line search. 

These numerical results meet the theoretical analysis as shown below: 
\begin{enumerate}
 \item Theorem
       \ref{theorem:BFGS-B-robustness} implies that the standard BFGS method is  
       robust against inexact line search. 
 \item As shown in Example \ref{example:VBFGS-power-div}, $V$-BFGS-B update with power
       potential is close to the standard BFGS update for large $n$ and moderate $\det(B_k)$. 
       That is, the mixing parameter $(s_k^\top y_k/s_k^\top B_k s_k)^\rho$ 
       in Example \ref{example:VBFGS-power-div} will be close to one if $n$ is large and 
       $s_k^\top y_k/s_k^\top B_k s_k$ does not depend on the dimension $n$ that much. 
       When $B_k$ has a large determinant which grows with the dimension $n$, 
       the number of $s_k^\top y_k/s_k^\top B_k s_k$ will severely depend on the dimension $n$. 
       Hence, the mixing parameter $(s_k^\top y_k/s_k^\top B_k s_k)^\rho$ will not close to one
       even for large $n$. Hence, in such case the influence function is affected by the
       choice of the power $\gamma$. 
       The same argument on the relation between influence function and the power $\gamma$ 
       will hold for the inverse Hessian update, that is, $V$-BFGS-H and $V$-DFP-H. 
 \item For $B_k=I+n^3pp^T$ the result on $V$-BFGS-B and $V$-DFP-B is numerically the same. 
       Under this setup, we can theoretically confirm that the influence functions of both update
       formula are identical to each other. On the other hand, 
       some calculation yields that the influence functions of $V$-BFGS-H and $V$-DFP-H
       are not the same. 
\end{enumerate}

The standard BFGS update formula achieves the min-max optimality of the gross error
sensitivity. That is, BFGS method may not be necessarily optimal for each setup. 
In numerical studies, however, BFGS method uniformly provides fairly stable update formula
compared to the other methods. 

\begin{table*}[tb]
 \footnotesize
\caption{Approximate influence function for $V$-BFGS update and $V$-DFP update is shown. 
 The power potential $V(z)=(1-z^\gamma)/\gamma$ is used for $V$-extended quasi-Newton
 methods, where $\gamma=0$ corresponds to BFGS or DFP method. }
 \label{tbl:gross_error_sensitivity}
\begin{center} 
\hspace*{-13mm}
\begin{tabular}{l|ccc|ccc|ccc}
 \multicolumn{10}{c}{$V$-BFGS-B}\\\hline
 \multicolumn{1}{c|}{$B_k$}
 &\multicolumn{3}{c|}{$\mathrm{diag}(1,\ldots,n)/(n!)^{1/n}$}
 &\multicolumn{3}{c|}{$\mathrm{diag}(1,\ldots,n)$}
 &\multicolumn{3}{c}{$I+n^3pp^\top$}\\ 
 \multicolumn{1}{c|}{$\gamma$}
         & $-2$  &  $-1$ &  $0$  &  $-2$ &  $-1$ &  $0$  &  $-2$ & $-1$  &  $0$  \\\hline
$n=10$   &  9.5e+00 & 9.5e+00 & 9.5e+00 & 1.5e+01 & 9.7e+00 & 9.5e+00 & 2.0e+02 & 1.0e+02 & 5.0e+01\\	   
$n=100$  &  2.7e+01 & 2.7e+01 & 2.7e+01 & 2.3e+02 & 2.8e+01 & 2.7e+01 & 1.1e+04 & 1.0e+04 & 8.7e+03\\	   
$n=500$  &  9.3e+01 & 9.3e+01 & 9.3e+01 & 2.8e+03 & 9.6e+01 & 9.3e+01 & 2.6e+05 & 2.5e+05 & 2.4e+05\\	   
$n=1000$ &  1.0e+02 & 1.0e+02 & 1.0e+02 & 7.4e+03 & 1.1e+02 & 1.0e+02 & 1.0e+06 & 9.9e+05 & 9.7e+05\\\hline
\end{tabular}								 
\vspace*{3mm}

\hspace*{-13mm}
\begin{tabular}{l|ccc|ccc|ccc}
 \multicolumn{10}{c}{$V$-DFP-B}\\\hline
 \multicolumn{1}{c|}{$B_k$}
 &\multicolumn{3}{c|}{$\mathrm{diag}(1,\ldots,n)/(n!)^{1/n}$}
 &\multicolumn{3}{c|}{$\mathrm{diag}(1,\ldots,n)$}
 &\multicolumn{3}{c}{$I+n^3pp^\top$}\\ 
 \multicolumn{1}{c|}{$\gamma$}
        & $-2$  &  $-1$&  $0$   &  $-2$  &  $-1$ & $0$   &  $-2$ &  $-1$ &  $0$  \\\hline
$n=10$  &  1.3e+02 & 1.3e+02 & 1.3e+02 & 2.9e+03 & 6.5e+02 & 1.5e+02 & 2.0e+02 & 1.0e+02 & 5.0e+01\\	   
$n=100$ &  1.7e+03 & 1.7e+03 & 1.7e+03 & 2.5e+06 & 6.5e+04 & 1.7e+03 & 1.1e+04 & 1.0e+04 & 8.7e+03\\	   
$n=500$ &  4.6e+04 & 4.6e+04 & 4.6e+04 & 1.6e+09 & 8.7e+06 & 4.7e+04 & 2.6e+05 & 2.5e+05 & 2.4e+05\\	   
$n=1000$&  3.0e+04 & 3.0e+04 & 3.0e+04 & 4.1e+09 & 1.1e+07 & 3.0e+04 & 1.0e+06 & 9.9e+05 & 9.7e+05\\\hline
\end{tabular}
\vspace*{3mm}

\hspace*{-13mm}
\begin{tabular}{l|ccc|ccc|ccc}
 \multicolumn{10}{c}{$V$-BFGS-H}\\\hline
 \multicolumn{1}{c|}{$H_k$}
 &\multicolumn{3}{c|}{$\mathrm{diag}(1,\ldots,n)/(n!)^{1/n}$}
 &\multicolumn{3}{c|}{$\mathrm{diag}(1,\ldots,n)$}
 &\multicolumn{3}{c}{$I+n^3pp^\top$}\\ 
 \multicolumn{1}{c|}{$\gamma$}
        & $-2$  &  $-1$&  $0$   &  $-2$  &  $-1$ & $0$   &  $-2$ &  $-1$ &  $0$  \\\hline
$n=10$  & 2.1e+02 & 2.1e+02 & 2.1e+02 & 4.8e+03 & 1.1e+03 & 2.4e+02 &2.2e+02 & 1.1e+02 & 5.6e+01\\	   
$n=100$ & 1.1e+03 & 1.1e+03 & 1.1e+03 & 1.6e+06 & 4.1e+04 & 1.1e+03 &2.0e+04 & 1.7e+04 & 1.5e+04\\	   
$n=500$ & 8.2e+04 & 8.2e+04 & 8.2e+04 & 2.8e+09 & 1.5e+07 & 8.3e+04 &8.7e+05 & 8.4e+05 & 8.1e+05\\	   
$n=1000$& 2.6e+04 & 2.6e+04 & 2.6e+04 & 3.6e+09 & 9.8e+06 & 2.7e+04 &4.7e+06 & 4.6e+06 & 4.5e+06\\\hline
 \end{tabular}
\vspace*{3mm}

\hspace*{-13mm}
\begin{tabular}{l|ccc|ccc|ccc}
 \multicolumn{10}{c}{$V$-DFP-H}\\\hline
 \multicolumn{1}{c|}{$H_k$}
 &\multicolumn{3}{c|}{$\mathrm{diag}(1,\ldots,n)/(n!)^{1/n}$}
 &\multicolumn{3}{c|}{$\mathrm{diag}(1,\ldots,n)$}
 &\multicolumn{3}{c}{$I+n^3pp^\top$}\\ 
 \multicolumn{1}{c|}{$\gamma$}
        & $-2$  &  $-1$&  $0$   &  $-2$  &  $-1$ & $0$   &  $-2$ &  $-1$ &  $0$  \\\hline
$n=10$  & 1.0e+01 & 1.0e+01 & 1.0e+01 & 1.7e+01 & 1.1e+01 & 1.0e+01 &2.5e+02 & 1.3e+02 & 6.4e+01\\	   
$n=100$ & 2.1e+01 & 2.1e+01 & 2.1e+01 & 4.5e+02 & 2.5e+01 & 2.1e+01 &4.1e+06 & 3.6e+06 & 3.1e+06\\	   
$n=500$ & 9.9e+01 & 9.9e+01 & 9.9e+01 & 9.5e+03 & 1.2e+02 & 9.9e+01 &1.4e+09 & 1.4e+09 & 1.3e+09\\	   
$n=1000$& 1.2e+02 & 1.2e+02 & 1.2e+02 & 3.6e+04 & 1.7e+02 & 1.2e+02 &1.2e+10 & 1.2e+10 & 1.2e+10\\\hline
 \end{tabular}
\end{center}
\end{table*}

Next, we apply the standard BFGS-B and DFP-B to solve the following two optimization problems: 
the quadratic convex problem
\begin{align*}
 \text{(Problem 1)}\qquad 
 \min_{x\in\Real^n}\ f(x)=\frac{1}{2}x^\top A x-e^\top x,
\end{align*}
where $e=(1,\ldots,1)^\top\in\Real^n$ and 
\begin{align*}
 A=\begin{pmatrix}
    2  & -1  &        &        & \\
    -1 & 2   & -1     &        & \\
       & -1  & 2      & \ddots & \\
       &     & \ddots & \ddots & -1\\
       &     &        & -1     & 2
   \end{pmatrix}\ \in\ \Real^{n\times n}, 
\end{align*}
and the boundary value problem \cite{fletcher95:_optim_posit_defin_updat_form} 
\begin{align*}
 \text{(Problem 2)}\qquad 
 \min_{x\in\Real^n}\ f(x)=\frac{1}{2}x^\top A x-e^\top
 x-\frac{1}{(n+1)^2}\sum_{i=1}^{n}(2x_i+\cos x_i),
\end{align*}
where the vector $e$ and the matrix $A$ are the same as problem 1. 
The objective function in problem 2 is non-linear and non-convex. 
The initial point $x_0$ is randomly 
generated by $n$-dimensional normal distribution with mean 
zero and variance-covariance matrix $10\times I$. 
The termination criterion 
\begin{align*}
 \|\nabla f(x_k)\|\leq n\times 10^{-5}\quad \text{or}\quad 
 k\geq 50000, 
\end{align*}
is employed, which is the same criterion used by Yamashita 
\cite{yamashita08:_spars_quasi_newton_updat_with}. 
Although the second criterion above implies that the method fails to obtain a solution, 
all trials did not reach the maximum number of iterations. 
In each problem, the step-length $\alpha_k$ is computed by the matlab
command ``{\tt fminbnd}'' with the option $\text{{\tt TolX}}=10^{-12}$ 
which denotes the termination tolerance on $x$. 
In the same way as the numerical studies on robustness of update formulae, 
the vector $s_k=x_{k+1}-x_{k}$ is
randomly perturbed such that $\widetilde{s}_k=(1+\varepsilon)s_k$, where $\varepsilon$ is 
a random variable according to the uniform distribution on the interval $[-h,h]$. 
The number of $h$ varies from $0$ to $0.3$. 
Accordingly, the vector $y_k=\nabla f(x_{k+1})-\nabla f(x_{k})$ is also 
changed to $\widetilde{y}_k=\nabla f(x_{k}+\widetilde{s}_k)-\nabla f(x_{k})$. 
As the result, for each iteration the secant condition with inexact line search is 
given as $B\widetilde{s}_k=\widetilde{y}_k$. 

The average number of iterations over $20$ runs for BFGS and DFP is shown in Table
\ref{tbl:num_iteration}. Compared to DFP method,  BFGS method requires fewer number of
iterations  to reach the optimal solution. Moreover, in BFGS update the number of
iterations is stable 
against the number of $h$. This result implies that BFGS is robust against random noise 
involved in inexact line search. 
On the other hand, the behaviour of DFP method is sensitive to contaminated step-length. 
Indeed, the number of iterations in DFP method rises drastically with the intensity of the
noise. For the quadratic convex objective function, the inexact line search does not
affect the secant condition. Hence the numerical result will impliy that the goodness of
the descent direction $B_k^{-1}\nabla f(x_k)$ in DFP will be easily degraded by inexact
line search. 
These numerical properties in quasi-Newton methods have been empirically
well-known \cite{conn88:_testin_class_of_algor_for,nocedal99:_numer_optim}. 
Powell \cite{powell86:_how_bad_are_bfgs_and} has theoretically studied the progression of
eigenvalues in approximate Hessian matrices in order to illustrate the difference between
BFGS and DFP. 

Through the numerical stduies in this section, we found that the theoretical framework
exploiting robust statistics can be a useful tool to investigate the property of
quasi-Newton methods.

\begin{table*}[tb]
 \small 
\caption{Number of iterations by BFGS and DFP under inexact line search. 
 The number of $h$ denotes intensity of noise involved in the line search.}
 \label{tbl:num_iteration}
\begin{center}
\begin{tabular}{llcccccc}
&   &  \multicolumn{2}{c}{$n=100$}
    &  \multicolumn{2}{c}{$n=500$}
    &  \multicolumn{2}{c}{$n=1000$} \\ 
           &  $h$ &  BFGS &   DFP  & BFGS  & DFP    & BFGS  & DFP \\ \hline
Problem 1  &  0.0 &  100.4 & 110.6 & 434.6 & 577.8  & 682.1 & 1788.5 \\ 
           &  0.1 &  102.9 & 166.2 & 430.6 &1165.2  & 680.9 & 2628.9 \\
           &  0.2 &  104.5 & 198.6 & 443.6 &1361.8  & 685.1 & 3099.2 \\
           &  0.3 &  106.0 & 223.0 & 444.2 &1501.6  & 687.6 & 3365.9 \\\hline
Problem 2  &  0.0 &  100.9 & 111.6 & 428.5 &585.7   & 661.5 & 2489.8 \\
           &  0.1 &  102.8 & 153.5 & 443.5 &1237.4  & 672.4 & 2762.1 \\
           &  0.2 &  104.4 & 177.7 & 438.3 &1419.6  & 682.7 & 3301.2 \\
           &  0.3 &  106.1 & 199.4 & 454.0 &1592.8  & 694.0 & 3730.8 \\\hline
\end{tabular}
\end{center}
\end{table*}

\section{Concluding Remarks}
\label{sec:Concluding_Remarks}
Along the line of the research stared by Fletcher
\cite{fletcher91:_new_resul_for_quasi_newton_formul}, we considered the quasi-Newton
update formula based on the Bregman divergence induced from potential functions. 
The proposed update formulae for the Hessian approximation belong to the class of
self-scaling quasi-Newton method. 
We studied the convergence property. Then, we applied the tools in the robust statistics
to analyze the robustness of the Hessian update formulae. 
As the result, we found that the influence of the inexact line search is bounded 
only for the standard BFGS formula for the Hessian approximation. 
Numerical studies support the usefulness of the theoretical framework borrowed from the
robust statistics. 

It will be an interesting future work to investigate the practical advantage of the 
self-scaling quasi-Newton methods derived from the $V$-Bregman divergence. 
Nocedal and Yuan proved that the self-scaling quasi-Newton method with 
the popular scaling parameter \eqref{eqn:popular-self-scaling-parameter} has some
drawbacks \cite{nocedal93:_analy_of_self_scalin_quasi_newton_method}. 
In our framework, the self-scaling quasi-Newton method with the scaling parameter
\eqref{eqn:popular-self-scaling-parameter} is out of the formulae 
derived from $V$-Bregman divergence. 
More precisely, the function $V(z)=n(1-z^{1/n})$, which is not potential, 
formally leads the popular self-scaling quasi-Newton formula. 
For the corresponding Bregman divergence $D_V(P,Q)$, 
the equality $D_V(P,cP)=0$ holds for any $P\in\PD(n)$ and any $c>0$. 
This property implies that the scale of the Hessian approximation is not fixed. 
We think that this property may lead some inefficiency of the self-scaling quasi-Newton
method with \eqref{eqn:popular-self-scaling-parameter}.  
The self-scaling quasi-Newton method associated with $V$-Bregman divergence may performs
well in practice. 

Another research direction is to consider the choice of the potential function $V$. 
Under the criterion of the gross error sensitivity, we found that the negative logarithmic
function $V(z)=-\log z$ is the optimal choice. The other criterion may lead other optimal 
potentials. 
Investigating the relation between the criterion for the update formula and the optimal
potential will be beneficial for the design of numerical algorithms.


\section{Acknowledgements}
The authors are grateful to Dr.~Nobuo Yamashita of Kyoto university for helpful comments. 
T.~Kanamori was partially supported by Grant-in-Aid for Young Scientists
(20700251).

\appendix 

\section{Proof of Theorems \ref{theorem:V-BFGS-form}}
\label{appendix:VBFGS-formula}
We prove the following lemma which is useful to show the existence of the optimal
solution. 
\begin{lemma}
 \label{lemma:sol_existence-nu-equation}
 Let $V$ be a potential and $\nu=\nu_V$. 
 For any $C>0$ the equation
 \begin{align}
  \label{eqn:lemma-equation-any-C}
 C\nu(z)^{n-1}=z,\quad z>0
 \end{align}
has the unique solution. 
\end{lemma}
\begin{proof}
 We define the function $\zeta(z)$ by $\zeta(z)=\log z-(n-1)\log \nu(z)$, 
 then, the \eqref{eqn:lemma-equation-any-C} is equivalent to the equation
 \begin{align}
  \label{eqn:expression2-equation-any-C}
  \log C=\zeta(z),\quad z>0. 
 \end{align}
 Since the potential function satisfies $\lim_{z\rightarrow+0}z/\nu(z)^{n-1}=0$ from the
 definition, we have $\lim_{z\rightarrow+0}\zeta(z)=-\infty$. 
 In terms of the derivative of $\zeta(z)$, we have the following inequality
 \begin{align*}
  \frac{d}{dz}\zeta(z)=\frac{1}{z}-(n-1)\frac{\beta(z)}{z}>\frac{1}{zn}>0. 
 \end{align*}
 Thus, $\zeta(z)$ is an increasing function on $\Real_+$. 
 Moreover we have
\begin{align*}
 \zeta(z)\geq \zeta(1)+\int_1^z\frac{1}{zn}dz=
 \zeta(1)+\frac{\log z}{n}. 
\end{align*}
 The above inequality implies that $\lim_{z\rightarrow\infty}\zeta(z)=\infty$. 
Since $\zeta(z)$ is continuous, the equation \eqref{eqn:expression2-equation-any-C} has
the unique solution. 
\end{proof}

\begin{proof}
[Proof of Theorem \ref{theorem:V-BFGS-form}]
First, we show the existence of the matrix $B_{k+1}$ 
satisfying \eqref{eqn:update-formula-V-BFGS}. Lemma \ref{lemma:sol_existence-nu-equation} now shows 
that there exists a solution $z^*>0$ for the equation 
\begin{align*}
 \frac{\det(B^{BFGS}[B_k;s_k,y_k])}{\nu(\det{B_k})^{n-1}}\cdot\nu(z)^{n-1}=z,\quad z>0. 
\end{align*}
By using the solution $z^*$, we define the matrix $\bar{B}$ such that
\begin{align*}
 \bar{B}=\frac{\nu(z^*)}{\nu(\det{B_k})}B^{BFGS}[B_k;s_k,y_k]
 +\big(1-\frac{\nu(z^*)}{\nu(\det{B_k})}\big)\frac{y_ky_k^\top}{s_k^\top y_k}, 
\end{align*}
then the determinant of $\bar{B}$ satisfies
\begin{align*}
 \det{\bar{B}}=\frac{\det(B^{BFGS}[B_k])}{\nu(\det{B_k})^{n-1}}\cdot \nu(z^*)^{n-1}=z^*, 
\end{align*}
in which the first equality comes from the formula 
$\det(A+vu^\top)= \det(A)(1+u^\top A^{-1}v)$ 
and the second one follows the definition of $z^*$. 
Hence there exists $B_{k+1}\in\mathrm{PD}(n)$ satisfying \eqref{eqn:update-formula-V-BFGS}. 

Next, we show that the matrix $B_{k+1}$ in \eqref{eqn:update-formula-V-BFGS} 
satisfies the optimality condition of \eqref{eqn:V-BFGS-hessian-update-prob}. 
According to G\"{u}ler, et al. \cite{guler09:_dualit_in_quasi_newton_method},
the normal vector for the affine subspace 
\[
\mathcal{M}=\{B\in\mathrm{PD}(n)~|~Bs_k=y_k\}
\]
is characterized by the form of 
\begin{align}
 \label{eqn:normal-vector-secant-cond}
 s_k\lambda^\top+\lambda s_k^\top\in\mathrm{Sym}(n),\qquad \lambda\in\Real^n. 
\end{align}
In fact for $B_1, B_2\in\mathcal{M}$ we have
\begin{align*}
  \<s_k\lambda^\top+\lambda s_k^\top,\,B_1-B_2\>
  &=
  \lambda^\top B_1 s_k+ s_k^\top B_1 \lambda
  -\lambda^\top B_2 s_k- s_k^\top B_2 \lambda\\
  &=
  \lambda^\top y_k+ y_k^\top \lambda
  -\lambda^\top y_k-y_k^\top \lambda\\
  &=0, 
\end{align*}
and thus $s_k\lambda^\top+\lambda s_k^\top$ is a normal vector of $\mathcal{M}$. 
G\"{u}ler, et al.~\cite{guler09:_dualit_in_quasi_newton_method} have shown that 
the normal vector is restricted to the form of \eqref{eqn:normal-vector-secant-cond}. 

Suppose $B'\in\mathrm{PD}(n)$ be an optimal solution of 
\eqref{eqn:V-BFGS-hessian-update-prob}, then $B'$ satisfies the optimality condition that
 there exists a vector $\lambda\in\Real^n$ such that 
\begin{align*}
 &\phantom{\Longleftrightarrow}
 \nabla_B D_V(B,B_k)\big|_{B=B'}=s_k\lambda^\top + \lambda s_k^\top
 \nonumber\\
 &\Longleftrightarrow\ 
 -\nu(\det(B'))(B')^{-1}+\nu(\det(B_k))B_k^{-1} 
 =s_k\lambda^\top + \lambda s_k^\top, 
\end{align*}
where $\nabla_B D_V(B,B_k)$ denotes the gradient of $D_V(B,B_k)$ with respect to the
variable $B$. Also, the optimal solution $B'$ should satisfy the constraint $B's_k=y_k$. 
On the other hand, the matrix $B_{k+1}$ defined by
 \eqref{eqn:update-formula-V-BFGS} satisfies
\begin{align*}
 B_{k+1}^{-1} 
 &=
 \frac{\nu(\det{B_k})}{\nu(\det{B_{k+1}})}(B^{BFGS}[B_k;s_k,y_k])^{-1}
 +\bigg(1-\frac{\nu(\det{B_k})}{\nu(\det{B_{k+1}})}\bigg)
 \frac{s_ks_k^\top}{s_k^\top y_k}\\
&=
 \frac{\nu(\det{B_k})}{\nu(\det{B_{k+1}})}B^{DFP}[B_k^{-1};y_k,s_k]
 +\bigg(1-\frac{\nu(\det{B_k})}{\nu(\det{B_{k+1}})}\bigg)
 \frac{s_ks_k^\top}{s_k^\top y_k}\\
\Longleftrightarrow &
 \left\{
 \begin{array}{l}
  \displaystyle  
   -\nu(\det{B_{k+1}})B_{k+1}^{-1}+
   \nu(\det{B_k})B_k^{-1}
   =s_k\lambda^\top+\lambda s_k^\top,\vspace*{2mm}\\
  \displaystyle  
   \lambda
   =\frac{\nu(\det{B_k})}{s_k^\top y_k}B_k^{-1}y_k
   -\frac{\nu(\det{B_{k+1}})}{2s_k^\top y_k}s_k
   -\frac{\nu(\det{B_k})y_k^\top B_k^{-1}y_k}{2(s_k^\top y_k)^2}s_k. 
 \end{array}
 \right.
\end{align*}
The conditions $s_k^\top y_k>0$ and $B_k\in\mathrm{PD}(n)$ guarantees 
the existence of the above vector $\lambda$. In addition, the direct computation yields
that the constraint $B_{k+1}s_k=y_k$ is satisfied. 
Hence, $B_{k+1}$ satisfies the optimality condition. 
Since \eqref{eqn:V-BFGS-hessian-update-prob} is a strictly convex problem, 
$B_{k+1}$ is the unique optimal solution. 
\end{proof}

\section{Proofs of Theorems \ref{theorem:existence-differentiability}}
\label{appendix:proof_theorem_existence_optsol}
We show that the optimal solution of $V$-BFGS-B is second order continuously 
differentiable. The same proof works for the other update formulae. 
\begin{proof}
We consider the problem \eqref{eqn:perturbed-V-BFGS}. Since the inequality 
$s^\top (y+\varepsilon\bar{y})>0$ holds for infinitesimal $\varepsilon$, 
Theorem \ref{theorem:V-BFGS-form} guarantees that there exists the unique optimal solution
 $B(\varepsilon)$ around $\varepsilon=0$. 
Let the function 
$F:\Real^{n\times n}\times \Real\rightarrow\Real^{n\times n}$ be 
\begin{align*}
 F(X,\varepsilon)
 &=\frac{1}{\nu(\det{X})}X-
\frac{1}{\nu(\det{B_k})}
 B^{BFGS}[B_k;(1+\varepsilon)s,y+\varepsilon\bar{y}]\\
 &\phantom{=}-\bigg(\frac{1}{\nu(\det{X})}
 -\frac{1}{\nu(\det{B_k})}\bigg)
 \frac{(y+\varepsilon\bar{y})(y+\varepsilon\bar{y})^\top}
 {(1+\varepsilon)s^\top(y+\varepsilon\bar{y})}, 
\end{align*}
for $X\in\Real^{n\times n}$ and $\varepsilon\in\Real$. 
For infinitesimal $\varepsilon$, the equality $F(B(\varepsilon),\varepsilon)=O$ holds, 
where $O$ is the null matrix. 
 We apply the implicit function theorem to prove the differentiability of
 $B(\varepsilon)$. 
 Since the potential function is third order continuously differentiable, 
 clearly $F(X,\varepsilon)$ is second order continuously differentiable in a vicinity of
 $(X,\varepsilon)=(B(0),0)$. For any symmetric matrix $A\in\mathrm{Sym}(n)$, the equality 
\begin{align*}
\nabla_X\<F(X,\,\varepsilon),\,A\>\big|_{X=B(0),\varepsilon=0}
=\frac{1}{\nu(\det{B(0)})}A-\frac{1}{\nu(\det{B(0)})^2}
 \bigg\<B(0)-\frac{yy^\top}{s^\top y},\,A\bigg\> B(0)^{-1}
\end{align*}
holds, where $\nabla_X$ denotes the gradient with respect to the variable $X$. This
 implies that the gradient of $F(X,\varepsilon)$ does not vanish at 
 $(X,\varepsilon)=(B(0),0)$. Hence, the implicit function theorem for $F(X,\varepsilon)$
 guarantees that $B(\varepsilon)$ is a second order continuously differentiable function
 with respect to  $\varepsilon$ in a vicinity of $\varepsilon=0$. 
\end{proof}

\section{Computations of Gross Error Sensitivity}
\label{appendix:robustness-sensitivity}
First, a universal formula for the computation of influence function is proved, and
some useful lemmas are prepared. Then, the gross error sensitivity for each update formula
is computed in Section 
\ref{appendix:proof-V-BFGS-B}, \ref{appendix:proof-V-DFP-B}, \ref{appendix:proof-V-BFGS-H}
and \ref{appendix:proof-V-DFP-H}. 
\begin{lemma}
\label{lemma:asympto-Beps}
Let $s,\bar{s},y$ and $\bar{y}$ be column vectors in $\Real^n$ such that $s^\top y>0$, and 
$B_k$ be a positive definite matrix. For an infinitesimal $\varepsilon$ let
 $B(\varepsilon)$ be the optimal solution of 
\begin{align}
 \label{lemma:appendix-opt-sol-perturbed-BFGS}
 \min_{B\in\mathrm{PD}(n)} D_V(B,B_k)\quad\subto 
 B(s+\varepsilon \bar{s})=y+\varepsilon \bar{y}, 
\end{align}
 and let $\Delta[B_k;s,\bar{s},y,\bar{y}]$ be the influence function
 $\dot{B}(0)$. Then we have 
 \begin{align}
   &\phantom{=}\dot{B}(0)\nonumber\\
  &=\Delta[B_k;s,\bar{s},y,\bar{y}]\nonumber\\
  &=
 \bigg\{
 \frac{s^\top\bar{y}-\bar{s}^\top y}{s^\top y}
 +\frac{\nu(\det{B(0)})}{\nu(\det{B_k})}
 \bigg( \frac{2\bar{s}^\top B_ks\cdot s^\top B_k(B(0))^{-1}B_ks}
  {(s^\top B_ks)^2} 
 -\frac{2\bar{s}^\top B_k(B(0))^{-1}B_ks}{s^\top B_ks} \bigg)
 \bigg\}\nonumber\\
 &\phantom{=} \times \frac{\beta(\det{B(0)})}{1-(n-1)\beta(\det{B(0)})}
  \bigg[B(0)-\frac{yy^\top}{s^\top y}\bigg]
 +\frac{y\bar{y}^\top+\bar{y}y^\top}{s^\top y}
 -\frac{s^\top\bar{y}+\bar{s}^\top y}{(s^\top y)^2}yy^\top\nonumber\\
 &\phantom{=} 
 +\frac{\nu(\det{B(0)})}{\nu(\det{B_k})}
 \bigg[
  \frac{2\bar{s}^\top B_ks}{(s^\top B_ks)^2} B_kss^\top B_k
  -\frac{B_k(s\bar{s}^\top+\bar{s}s^\top)B_k}{s^\top B_ks}
 \bigg]. 
  \label{eqn:Delta-expansion}
 \end{align}
\end{lemma}
The matrix $\Delta[B_k;s,\bar{s},y,\bar{y}]$ is well-defined, since the inequalities
$\nu>0$ and $1-(n-1)\beta>0$ hold for any potential function. 
Note that $\Delta[B_k;s,s,y,y]=O$ holds. This is another proof of Lemma 
\ref{lemma:quadratic_function_grosserror}. 
\begin{proof}[Proof of Lemma \ref{lemma:asympto-Beps}]
 In the same way as the proof of Theorem \ref{theorem:V-BFGS-form} and Theorem
 \ref{theorem:existence-differentiability}, we can prove the existence and the
 differentiability of $B(\varepsilon)$. Since $B(\varepsilon)$ is second order
 continuously differentiable around $\varepsilon=0$, the equality 
\begin{align*}
 B(\varepsilon)=B(0)+\varepsilon\Delta+O(\varepsilon^2), 
\end{align*}
 holds, where $\Delta\in\mathrm{Sym}(n)$. Then we have
\begin{align*}
 \det(B(\varepsilon))
 &=\det(B(0)+\varepsilon\Delta+O(\varepsilon^2))\\
 &=\det(B(0))+\varepsilon\det(B(0))\<\Delta, B(0)^{-1}\>+O(\varepsilon^2)
\end{align*}
and thus we obtain
\begin{align*}
 \nu(\det{B(\varepsilon)})
 &=
 \nu(\det{B(0)})+\varepsilon\nu'(\det{B(0)})
 \det(B(0))\<\Delta, B(0)^{-1}\>+O(\varepsilon^2). 
\end{align*}
For simplicity let $\delta$ be
\begin{align}
 \label{eqn:delta-Delta}
 \delta=\det(B(0))\<\Delta, B(0)^{-1}\>
\end{align}
then the equality
\begin{align}
\label{eqn:nu-det-asympt}
 \nu(\det{B(\varepsilon)})
 &=
 \nu(\det{B(0)})+\varepsilon\cdot
 \delta\cdot\nu'(\det{B(0)})+O(\varepsilon^2)
\end{align}
holds. 
By some calculation, we see that the asymptotic expansion of 
$B^{BFGS}[B_k;s+\varepsilon \bar{s},\,y+\varepsilon\bar{y}]$ 
and $(y+\varepsilon\bar{y})(y+\varepsilon\bar{y})^\top/
(s+\varepsilon s)^\top(y+\varepsilon\bar{y})$
are respectively given by 
\begin{align}
 &\phantom{=}B^{BFGS}[B_k;s+\varepsilon \bar{s},y+\varepsilon\bar{y}]
 \nonumber\\
 &=
 B^{BFGS}[B_k;s,y]  \nonumber\\
 &\phantom{=}+ \varepsilon\bigg(
 \frac{y\bar{y}^\top+\bar{y}y^\top}{s^\top y}
 -\frac{s^\top\bar{y}+\bar{s}^\top y}
 {(s^\top y)^2} yy^\top 
 -\frac{B_k(s\bar{s}^\top+\bar{s}s^\top)B_k}{s^\top B_ks}
 +\frac{2\bar{s}^\top B_ks}{(s^\top B_ks)^2}B_kss^\top B_k\bigg)
 \nonumber\\
 &\phantom{=}+O(\varepsilon^2)
 \label{eqn:BFGS-asympt}
\end{align}
and 
\begin{align}
 \frac{(y+\varepsilon\bar{y})(y+\varepsilon\bar{y})^\top}
 {(s+\varepsilon\bar{s})^\top(y+\varepsilon\bar{y})}
 &=
 \frac{yy^\top}{s^\top y}+\varepsilon
 \left(
 \frac{y\bar{y}^\top+\bar{y}y^\top}{s^\top y}
 -\frac{s^\top\bar{y}+\bar{s}^\top y}{(s^\top y)^2} yy^\top
 \right)
 +O(\varepsilon^2). 
 \label{eqn:yy-asympt} 
\end{align}
Substituting \eqref{eqn:nu-det-asympt}, \eqref{eqn:BFGS-asympt} and \eqref{eqn:yy-asympt}
into the equality 
\begin{align*}
 B(\varepsilon)
 =
\frac{\nu(\det{B(\varepsilon)})}{\nu(\det{B_k})}
 B^{BFGS}[B_k;s+\varepsilon \bar{s},y+\varepsilon\bar{y}]
 +
 \bigg(1-\frac{\nu(\det{B(\varepsilon)})}{\nu(\det{B_k})}\bigg)
 \frac{(y+\varepsilon\bar{y})(y+\varepsilon\bar{y})^\top}
 {(s+\varepsilon \bar{s})^\top(y+\varepsilon\bar{y})}, 
\end{align*}
we obtain
\begin{align*}
 &\phantom{=}B(\varepsilon)\\
 &=
 B(0) +\varepsilon\cdot
 \bigg\{
 \delta\cdot\frac{\nu'(\det{B(0)})}{\nu(\det{B_k})}
 \big( B^{BFGS}[B_k;s,y]-\frac{yy^\top}{s^\top y}\big)
 +\frac{y\bar{y}^\top+\bar{y}y^\top}{s^\top y}
 -\frac{s^\top\bar{y}+\bar{s}^\top y}{(s^\top y)^2}yy^\top\\
 &\phantom{=}
 -\frac{\nu(\det{B(0)})}{\nu(\det{B_k})}
 \frac{B_k(s\bar{s}^\top+\bar{s}s^\top)B_k}{s^\top B_ks}
 +\frac{\nu(\det{B(0)})}{\nu(\det{B_k})}
 \frac{2\bar{s}^\top B_ks}{(s^\top B_ks)^2} B_kss^\top B_k
 \bigg\}
 +O(\varepsilon^2), 
\end{align*}
and thus $\Delta$ is represented as
\begin{align*}
 \Delta
 &= 
 \delta\cdot\frac{\nu'(\det{B(0)})}{\nu(\det{B_k})}
 \bigg[ B^{BFGS}[B_k;s,y]-\frac{yy^\top}{s^\top y}\bigg]
 +\frac{y\bar{y}^\top+\bar{y}y^\top}{s^\top y}
 -\frac{s^\top\bar{y}+\bar{s}^\top y}{(s^\top y)^2}yy^\top\\
 &\phantom{=}
 -\frac{\nu(\det{B(0)})}{\nu(\det{B_k})}
 \frac{B_k(s\bar{s}^\top+\bar{s}s^\top)B_k}{s^\top B_ks}
 +\frac{\nu(\det{B(0)})}{\nu(\det{B_k})}
 \frac{2\bar{s}^\top B_ks}{(s^\top B_ks)^2} B_kss^\top B_k\\
 &= 
 \delta\cdot
 \frac{\nu'(\det{B(0)})}{\nu(\det{B(0)})}
 \big[B(0)-\frac{yy^\top}{s^\top y}\big]
 +\frac{y\bar{y}^\top+\bar{y}y^\top}{s^\top y}
 -\frac{s^\top\bar{y}+\bar{s}^\top y}{(s^\top y)^2}yy^\top\\
 &\phantom{=}
 -\frac{\nu(\det{B(0)})}{\nu(\det{B_k})}
 \frac{B_k(s\bar{s}^\top+\bar{s}s^\top)B_k}{s^\top B_ks}
 +\frac{\nu(\det{B(0)})}{\nu(\det{B_k})}
 \frac{2\bar{s}^\top B_ks}{(s^\top B_ks)^2} B_kss^\top B_k
\end{align*}
in which we use the equality
\begin{align*}
 \frac{\nu(\det{B(0)})}{\nu(\det{B_k})}
 \bigg[ B^{BFGS}[B_k;s,y]-\frac{yy^\top}{s^\top y}\bigg]
 = B(0)-\frac{yy^\top}{s^\top y}. 
\end{align*}
Substituting the above $\Delta$ into \eqref{eqn:delta-Delta}, we have
\begin{align*}
 \delta
 &=
 \frac{\det{B(0)}}{1-\beta(\det{B(0)})(n-1)}
 \bigg\{
 \frac{s^\top\bar{y}-\bar{s}^\top y}{s^\top y}\\
 &\phantom{=}\qquad
 +\frac{\nu(\det{B(0)})}{\nu(\det{B_k})}
 \bigg( \frac{2\bar{s}^\top B_ks\cdot s^\top B_k(B(0))^{-1}B_ks}{(s^\top B_ks)^2} 
 -\frac{2\bar{s}^\top B_k(B(0))^{-1}B_ks}{s^\top B_ks} \bigg)
 \bigg\}. 
\end{align*}
 As the result, we obtain
\begin{align*}
 &\phantom{=}\frac{B(\varepsilon)-B(0)}{\varepsilon}\\
 &=
 \bigg\{
 \frac{s^\top\bar{y}-\bar{s}^\top y}{s^\top y}
 +\frac{\nu(\det{B(0)})}{\nu(\det{B_k})}
 \bigg( \frac{2\bar{s}^\top B_ks\cdot s^\top B_k(B(0))^{-1}B_ks}
 {(s^\top B_ks)^2} 
 -\frac{2\bar{s}^\top B_k(B(0))^{-1}B_ks}{s^\top B_ks} \bigg)
 \bigg\}\nonumber\\
 &\ \times \frac{\beta(\det{B(0)})}{1-(n-1)\beta(\det{B(0)})}
\bigg[B(0)-\frac{yy^\top}{s^\top y}\bigg]
 +\frac{y\bar{y}^\top+\bar{y}y^\top}{s^\top y}
 -\frac{s^\top\bar{y}+\bar{s}^\top y}{(s^\top y)^2}yy^\top\nonumber\\
 &\phantom{=} 
 +\frac{\nu(\det{B(0)})}{\nu(\det{B_k})}
 \bigg[
 \frac{2\bar{s}^\top B_ks}{(s^\top B_ks)^2} B_kss^\top B_k
 -\frac{B_k(s\bar{s}^\top+\bar{s}s^\top)B_k}{s^\top B_ks}
 \bigg]+O(\varepsilon). 
\end{align*}
Letting $\varepsilon$ tend to zero, we obtain the influence function 
 $\dot{B}(0)=\Delta[B_k;s,\bar{s},y,\bar{y}]$.
\end{proof}

\begin{lemma}
\label{lemma:asympto-Heps}
Let $s,\bar{s},y$ and $\bar{y}$ be a set of column vectors in $\Real^n$ such
that $s^\top y>0$ and $B_k$ be a matrix in $\mathrm{PD}(n)$.  
For an infinitesimal $\varepsilon$ let $B(\varepsilon)$ be the optimal solution of 
\begin{align*}
 \min_{B\in\mathrm{PD}(n)} D_V(B^{-1},B_k^{-1})\quad\subto 
 B(s+\varepsilon \bar{s})=y+\varepsilon \bar{y}
\end{align*}
and let $\Gamma[B_k;s,\bar{s},y,\bar{y}]$ be $\dot{B}(0)$ then we have 
\begin{align}
 \label{eqn:asympto-gamma}
\Gamma[B_k;s,\bar{s},y,\bar{y}]
 = -B(0) \Delta[B_k^{-1};y,\bar{y},s,\bar{s}] B(0), 
\end{align}
where $\Delta$ is the function defined in Lemma \ref{lemma:asympto-Beps}. 
\end{lemma}
\begin{proof}
 Let $H(\varepsilon)$ be the optimal solution of 
  \begin{align*}
  \min_{H\in\mathrm{PD}(n)} D_V(H,B_k^{-1})\quad\subto 
   H(y+\varepsilon \bar{y})=s+\varepsilon \bar{s}
 \end{align*}
 then, clearly $B(\varepsilon)=H(\varepsilon)^{-1}$ holds. 
 Thus we have
\begin{align*}
\Gamma[B_k;s,\bar{s},y,\bar{y}]
 &=
 \dot{B}(0)
 = -H(0)^{-1}\dot{H}(0)H(0)^{-1}
 = -B(0)\Delta[B_k^{-1};y,\bar{y},s,\bar{s}]B(0), 
\end{align*}
where $\dot{H}(0)=\Delta[B_k^{-1};y,\bar{y},s,\bar{s}]$ is applied. 
\end{proof}
We show another lemma which is useful to prove that the gross error sensitivity 
diverges to infinity. 
\begin{lemma}
 \label{lemma:B_i-existence-lemma}
 Suppose $n\geq k+3$ for non-negative integers $n$ and $k$. 
 For any set of vectors $s,y,y_1\ldots,y_k\in\Real^n$ 
 such that $s^\top y>0$ and any positive real number $d$, there 
 exists a sequence $\{B_i\}_{i=1}^\infty\subset\mathrm{PD}(n)$ satisfying the following
 three conditions: 
 \begin{enumerate}
  \item The equalities $B_iy=s$ 
	and $B_iy_m=B_jy_m$ hold for all $i,j\geq 1$ and $m=1,\ldots,k$. 
  \item $\det(B_i)=d$ for all $i\geq 1$. 
  \item $\lim_{i\rightarrow\infty}\|B_i\|_F=\infty$. 
 \end{enumerate}
\end{lemma}
\begin{proof}
 For any $s,y\in\Real^n$ such that $s^\top y>0$ there exists
 $\bar{B}\in\mathrm{PD}(n)$ satisfying $\bar{B}s=y$. 
 Indeed, for the $n$ by $n$ identity matrix $I$, 
 the matrix $\bar{B}=B^{BFGS}[I;s,y]\in\mathrm{PD}(n)$ is well-defined and satisfies 
 $\bar{B}s=y$. 
 When $n\geq k+3$ holds, there exist two unit vectors $p_1, p_2\in\Real^n$ satisfying 
 $p_1^\top p_2=0$ and 
  \begin{align*}
  &p_1^\top (\bar{B}^{1/2} s)=0,\quad
  p_1^\top (\bar{B}^{1/2} y_m)=0,  \quad m=1,\ldots,k,\\
  &p_2^\top (\bar{B}^{1/2} s)=0,\quad 
  p_2^\top (\bar{B}^{1/2} y_m)=0,  \quad m=1,\ldots,k. 
 \end{align*}
 We will show that the matrix 
\begin{align*}
 B(a)=\bar{B}^{1/2}(I+a p_1p_1^\top+bp_2p_2^\top)\bar{B}^{1/2}
\end{align*}
 with 
 \begin{align}
  \label{eqn:lemma-ab-condition}
  a>0,\qquad b=\frac{d/\det(\bar{B})}{1+a}-1
 \end{align}
 satisfies four conditions: $B(a)s=y,\,B(a)y_m=\bar{B}y_m,\,\det{B(a)}=d$ and 
 $B(a)\in\mathrm{PD}(n)$ for all $a>0$. 
 The first two equalities are clear from the definition of $p_1,\,p_2$ and $\bar{B}$. 
 The determinant of $B(a)$ is equal to
 \begin{align*}
  \det(B(a))=\det(\bar{B})\det(I+ap_1p_1^\top+bp_2p_2^\top)=\det(\bar{B})(1+a)(1+b)=d. 
 \end{align*}
 For any unit vector $x\in\Real^n$ we have
 \begin{align*}
  x^\top (I+ap_1p_1^\top+bp_2p_2^\top) x
  &=
  1+a(p_1^\top x)^2+b(p_2^\top x)^2\\
  &\geq 
  1+b(p_2^\top x)^2\qquad(\because a>0)\\
  &\geq 
  1-(p_2^\top x)^2\qquad\ (\because b>-1)\\
  &\geq 0\qquad (\text{Schwarz inequality})
 \end{align*}
 and in addition the determinant of $(I+ap_1p_1^\top+bp_2p_2^\top)$ is equal to
 $d/\det(\bar{B})>0$.   
 Thus $B(a)$ is positive definite. 
 Let $\lambda_1(a)$ be the maximum eigenvalue of $B(a)$, and 
 $x$ be a unit vector defined by $x=\bar{B}^{-1/2}p_1/\|\bar{B}^{-1/2}p_1\|$. 
 Then in terms of the maximum eigenvalue of $B(a)$ we have
 \begin{align*}
  \|B(a)\|_F\geq \lambda_1(a)\geq 
  x^\top \bar{B}x+\frac{a}{p_1^\top \bar{B}^{-1}p_1}. 
 \end{align*}
 Then $\|B(a)\|_F$ tends to infinity when $a$ tends to infinity. 
 Thus the sequence defined by 
 \begin{align}
  B_i=B(i),\quad i=1,2,3,\ldots
 \end{align}
 satisfies the conditions of the lemma. 
\end{proof}

\subsection{Proof of Theorem \ref{theorem:BFGS-B-robustness}}
\label{appendix:proof-V-BFGS-B}
Let $B(\varepsilon)$ be the optimal solution of \eqref{eqn:perturbed-V-BFGS}. 
Under the inexact line search, the influence function $\dot{B}(0)$ for $V$-BFGS-B is equal
to $\Delta[B_k;s,s,y,\bar{y}]$ which is defined in Lemma \ref{lemma:asympto-Beps}. 
Thus we have
\begin{align}
 \dot{B}(0) =
\frac{(\bar{y}-y)^\top s}{s^\top y}
 \frac{\beta(\det{B(0)})}{1-(n-1)\beta(\det{B(0)})}
 \bigg[B(0)-\frac{yy^\top}{s^\top y}\bigg]
 +\frac{y\bar{y}^\top+\bar{y}y^\top}{s^\top y}
 -\frac{(y+\bar{y})^\top s}{(s^\top y)^2}yy^\top. 
 \label{eqn:BFGS-B-perturb-xx}
\end{align}
 If $(\bar{y}-y)^\top s=0$ holds for any $\bar{y}\in{\mathcal{Y}}$, 
 the potential does not affect the norm of the influence function, 
 because the first term of the above expression vanishes. 
 Thus, clearly $V(z)=-\log(z)$ is an optimal potential. 
 Below we assume $(\bar{y}-y)^\top s\neq 0$ for a vector $\bar{y}\in{\mathcal{Y}}$. 
 Suppose that $B_k$ satisfies $B_ks=y$. Then $B(0)=B_k$  holds, and the triangle
 inequality yields that 
 \begin{align*}
  \|\dot{B}(0)\|_F&=\|\Delta[B_k;s,s,y,\bar{y}]\|_F\\
  &\geq 
  \bigg|\frac{(\bar{y}-y)^\top s}{s^\top y}\bigg|
  \bigg|\frac{\beta(\det{B_k})}{1-(n-1)\beta(\det{B_k})}\bigg|
  \bigg(
  \|B_k\|_F-\big\|\frac{yy^\top}{s^\top y}\big\|_F  
  \bigg)  \\
  &\phantom{=}
  -\bigg\|
  \frac{y\bar{y}^\top+\bar{y}y^\top}{s^\top y}
  -\frac{(\bar{y}+y)^\top s}{(s^\top y)^2}yy^\top
  \bigg\|_F. 
 \end{align*}
 If $\beta(z)$ is not the null function, there exists $d>0$ 
 such that $\beta(d)\neq 0$. Lemma \ref{lemma:B_i-existence-lemma} with $k=0$ implies that
 for $n\geq 3$ there exists a sequence $\{\bar{B}_i\}\subset\mathrm{PD}(n)$ satisfying 
 $\bar{B}_is=y,\,\det{\bar{B}_i}=d$ for all $i$ and
 $\lim_{i\rightarrow\infty}\|\bar{B}_i\|_F=\infty$. 
 Hence 
 \begin{align*}
  \lim_{i\rightarrow\infty}\|\Delta[\bar{B}_i;s,s,y,\bar{y}]\|_F=\infty
 \end{align*}
 holds, and then we obtain
\begin{align*}
  \sup\{~\|\Delta[B_k;s,s,y,\bar{y}]\|_F
  ~|~B_k\in\mathrm{PD}(n),\,\bar{y}\in{\mathcal{Y}}~\} 
 =\infty.
\end{align*}
On the other hand, if $\beta(z)=0$ for all $z>0$, we obtain
 \begin{align*}
  \max_{B_k,\bar{y}}\|\Delta[B_k;s,s,y,\bar{y}]\|_F
  =
  \max_{\bar{y}\in {\mathcal{Y}}}
  \bigg\|
  \frac{y\bar{y}^\top+\bar{y}y^\top}{s^\top y}
  -\frac{(\bar{y}+y)^\top s}{(s^\top y)^2}yy^\top
  \bigg\|_F<\infty, 
 \end{align*}
 since ${\mathcal{Y}}$ is bounded. As the result, the potential $V$ such that
 $\beta_V=0$ minimizes the gross error sensitivity. 
 The condition $\beta_V=0$ leads to $V(z)=-\log(z)$ up to a constant factor.

\subsection{Proof of Theorem \ref{theorem:DFP-B-robustness}}
\label{appendix:proof-V-DFP-B}
Let $B(\varepsilon)$ be the optimal solution of \eqref{eqn:perturbed-V-DFP}. 
Under the inexact line search, 
the influence function $\dot{B}(0)$ for $V$-DFP-B is equal to $\Gamma[B_k;s,s,y,\bar{y}]$
which is defined in Lemma \ref{lemma:asympto-Heps}. 

First, we study the case that $\beta(z)$ is not the null function. 
For the matrix $B_k$ such that $B_ks=y$, we have $B(0)=B_k$. 
Using Lemma \ref{lemma:asympto-Heps} for $B(0)=B_k$, we have 
\begin{align*}
 \phantom{=}\dot{B}(0)
 &=-B_k\Delta[B_k^{-1};y,\bar{y},s,s]B_k\\
 &=
 \frac{(\bar{y}-y)^\top s}{s^\top y}\cdot
 \frac{\beta(\det(B_k)^{-1})}{1-(n-1)\beta(\det(B_k)^{-1})}
 \bigg[B_k-\frac{yy^\top}{s^\top y}\bigg]
 +\frac{y\bar{y}^\top+\bar{y}y^\top}{s^\top y}
 -\frac{(y+\bar{y})^\top s}{(s^\top y)^2}yy^\top, 
\end{align*}
in which the equality $B_ks=y$ is used. 
The above expression is almost same as \eqref{eqn:BFGS-B-perturb-xx} with $B(0)=B_k$, and
thus the same proof works to obtain
\begin{align*}
 \sup\{~\|\dot{B}(0)\|_F~|~B_k\in\mathrm{PD}(n),\,\bar{y}
 \in{\mathcal{Y}}~\}=\infty. 
\end{align*}

Next, we study the case that $\beta$ is the null function, that is, $\beta(z)=0$. 
Then, $V(z)=-\log(z)$ and $\nu(z)=1$ hold. 
Let $B_k$ be a positive definite matrix which does not necessarily satisfy $B_ks=y$. 
Then we obtain 
\begin{align*}
\dot{B}(0)
 &=-B(0)\Delta[B_k^{-1};y,\bar{y},s,s]B(0)\\
 &=
 -\frac{(y-\bar{y})^\top s}{(s^\top y)^2} yy^\top +
 \frac{B(0)B_k^{-1}(y\bar{y}^\top+\bar{y}y^\top) B_k^{-1}B(0)}{y^\top B_k^{-1}y}\\
 &\phantom{=}-\frac{2\bar{y}^\top B_k^{-1}y}{(y^\top B_k^{-1}y)^2}
 B(0)B_k^{-1}yy^\top B_k^{-1}B(0)
\end{align*}
in which we used $B(0)s=y$. For $\beta=0$, the updated matrix $B(0)$ is equal to
$B^{DFP}[B_k;s,y]$ and thus, we have
\begin{align}
 \label{eqn:appendix-B_DFP0_x_invB_k}
 B(0)B_k^{-1}=
 I-\frac{B_ksy^\top B_k^{-1}+ys^\top}{s^\top y}+\frac{s^\top B_k s}{(s^\top y)^2} yy^\top B_k^{-1}
 +\frac{1}{s^\top y}yy^\top B_k^{-1}. 
\end{align}
Let $\bar{B}\in\mathrm{PD}(n)$ and $c$ be a positive real number, and we define
$t=\bar{B}s$, then for $B_k=c\bar{B}$ some calculation yields 
\begin{align*}
\dot{B}(0)
 =-B(0)\Delta[(c\bar{B})^{-1};y,\bar{y},s,s]B(0)
 = -\frac{c}{s^\top y} Z
 -\frac{(y+\bar{y})^\top s}{(s^\top y)^2}yy^\top
 +\frac{y\bar{y}^\top+\bar{y}y^\top}{s^\top y},
\end{align*}
where $Z$ is defined by 
\begin{align*}
 Z&= \bigg(t-\frac{s^\top t}{s^\top y}y\bigg)
 \bigg(\bar{y}-\frac{s^\top \bar{y}}{s^\top y}y\bigg)^\top
 +\bigg(\bar{y}-\frac{s^\top \bar{y}}{s^\top y}y\bigg)
 \bigg(t-\frac{s^\top t}{s^\top y}y\bigg)^\top. 
\end{align*}
Since ${\mathcal{Y}}$ contains an open subset, 
there exists a vector $\bar{y}\in{\mathcal{Y}}$ which is linearly independent to $y$. 
Clearly there exists $\bar{B}\in\mathrm{PD}(n)$ such that three vectors,
$t=\bar{B}s,\,\bar{y}$ and $y$, are linearly independent.  
For such choice, $Z$ is not the null matrix, and the equality
 \begin{align*}
 \lim_{c\rightarrow\infty}
 \|B(0)\Delta[(c\bar{B})^{-1};y,\bar{y},s,s]B(0)\|_F=\infty
\end{align*}
holds. As the result, even for the standard DFP formula, we have
\begin{align*}
 \sup\big\{\|\dot{B}(0)\|_F~|~ B\in\mathrm{PD}(n),\,
 \bar{y}\in{\mathcal{Y}}\big\}=\infty. 
\end{align*}
In summary, for all $V$-DFP update for the Hessian approximation, the gross error 
sensitivity defined in Theorem \ref{theorem:DFP-B-robustness} is equal to infinity.

\subsection{Proof of Theorem \ref{theorem:BFGS-H-robustness}}
\label{appendix:proof-V-BFGS-H}
Let $H(\varepsilon)$ be the optimal solution of \eqref{eqn:perturbed-V-BFGS-H}. 
Under the inexact line search, the influence function $\dot{H}(0)$ for $V$-BFGS-H is equal
to $\Gamma[H_k;y,\bar{y},s,s]$ which is defined in Lemma \ref{lemma:asympto-Heps}. 

First, we study the case that $\beta(z)$ is not the null function. 
Suppose $\beta(d)\neq 0$. If $H_k$ satisfies $H_ky=s$, then we have $H_k=H(0)$. 
Using Lemma \ref{lemma:asympto-Beps} and Lemma \ref{lemma:asympto-Heps} for the matrix
$H_k$ such that $H_ky=s$, we obtain 
\begin{align}
 &\phantom{=}\dot{H}(0) \nonumber\\
 &= -H_k \Delta[H_k^{-1};s,s,y,\bar{y}]H_k\nonumber\\
 &=
 \frac{(y-\bar{y})^\top s}{s^\top y}
 \frac{\beta(\det(H_k)^{-1})}{1-(n-1)\beta(\det(H_k)^{-1})}
 \bigg[H_k-\frac{ss^\top}{s^\top y}\bigg]
 -\frac{H_k\bar{y}s^\top +s\bar{y}^\top H_k}{s^\top y}
 +\frac{(y+\bar{y})^\top s}{(s^\top y)^2}ss^\top. 
 \label{eqn:perturb-BFGS-H-beta_nonzero}
\end{align}
Lemma \ref{lemma:B_i-existence-lemma} with $k=1$ implies that
for $n\geq 4$ there exists a sequence $\{\bar{H}_i\}\subset\mathrm{PD}(n)$ satisfying the
following conditions:
$\bar{H}_iy=s$ and $(\det{\bar{H}_i})^{-1}=d$ for all $i\geq 1$; 
$\bar{H}_i\bar{y}=\bar{H}_j\bar{y}$ for all $i,j\geq 1$;
$\lim_{i\rightarrow\infty}\|\bar{H}_i\|_F=\infty$. 
We define $\bar{t}=\bar{H}_i\bar{y}$ which does not depend on $i$. 
Then for $H_k=\bar{H}_i$ we have 
\begin{align*}
 &\phantom{=}\|\dot{H}(0)\|_F\\
 &=\|\bar{H}_i \Delta[\bar{H}_i^{-1},s,s,y,\bar{y}]\bar{H}_i\|_F\\
 &\geq
 \left|\frac{(y-\bar{y})^\top s}{s^\top y} \right| 
 \left| \frac{\beta(d)}{1-(n-1)\beta(d)}\right| 
 \bigg(\|\bar{H}_i\|-\big\|\frac{ss^\top}{s^\top y}\big\|\bigg)
 -
 \bigg\| \frac{\bar{t}s^\top +s\bar{t}^\top}{s^\top y}
 -\frac{(y+\bar{y})^\top s}{(s^\top y)^2}ss^\top \bigg\|. 
\end{align*}
 Hence the equality 
 \begin{align*}
  \lim_{i\rightarrow\infty}
  \|\bar{H}_i \Delta[\bar{H}_i^{-1},s,s,y,\bar{y}]\bar{H}_i\|_F=\infty
 \end{align*}
 holds, and thus we obtain
\begin{align*}
\sup\big\{~ \|\dot{H}(0)\|_F ~|~H_k\in\mathrm{PD}(n),\,\bar{y}\in{\mathcal{Y}}~\big\}
 =\infty. 
\end{align*}

Next, we study the case that $\beta$ is the null function, that is, $\beta(z)=0$. 
Then, $V(z)=-\log(z)$ and $\nu(z)=1$ holds. For $H_k$ such that $H_ky=s$, we have
\begin{align}
 \label{eqn:standard-BFGS-H-influence-func}
 \dot{H}(0)=
 -\frac{H_k\bar{y}s^\top +s\bar{y}^\top H_k}{s^\top y}
 +\frac{(y+\bar{y})^\top s}{(s^\top y)^2}ss^\top. 
\end{align}
Let $\bar{H}_0\in\mathrm{PD}(n)$ be a matrix satisfying $\bar{H}_0y=s$. 
Let $p_1\in\Real^n$ and $\bar{y}\in{\mathcal{Y}}$ be vectors satisfying 
$p_1^\top \bar{H}_0^{1/2}y=0$ and $p_1^\top \bar{H}_0^{1/2}\bar{y}\neq 0$. 
For $n\geq 4$, the existence of $p_1$ and $\bar{y}$ is guaranteed by the assumption on
${\mathcal{Y}}$. Indeed, there exists $\bar{y}\in{\mathcal{Y}}$ such that $\bar{y}$ and
$y$ are linearly independent. We now define the matrix $\bar{H}_i\in\mathrm{PD}(n)$ by
\begin{align*}
 \bar{H}_i=\bar{H}_0^{1/2}(I+i\cdot p_1p_1^\top)\bar{H}_0^{1/2},\quad i=0,1,2,\ldots
\end{align*}
Then we have
\begin{align*}
 &\bar{H}_iy=s, \quad 
 \bar{H}_i\bar{y}= z+i\cdot u, 
\end{align*}
where 
$z=\bar{H}_0\bar{y}$ and 
$u= (p_1^\top\bar{H}_0^{1/2}\bar{y})\bar{H}_0^{1/2}p_1\neq 0$. 
Substituting $H_k=\bar{H}_i$ into \eqref{eqn:standard-BFGS-H-influence-func}, we obtain 
\begin{align*}
 \dot{H}(0)
 &=
 -i\cdot\frac{us^\top +su^\top}{s^\top y}
 +\frac{(y+\bar{y})^\top s}{s^\top y}ss^\top
 -\frac{zs^\top +sz^\top}{s^\top y}. 
\end{align*}
This implies that 
\begin{align*}
 \lim_{i\rightarrow\infty}\|\bar{H}_i\Delta[\bar{H}_i^{-1};s,s,y,\bar{y}]\bar{H}_i\|=\infty. 
\end{align*}
for $\beta=0$. Hence we obtain
\begin{align*}
  \sup\big\{\|\dot{H}(0)\|_F ~|~H_k\in\mathrm{PD}(n),\,\bar{y}\in{\mathcal{Y}}\big\} 
 =\infty
\end{align*}
even for the standard BFGS update of the inverse Hessian approximation.

\subsection {Proof of Theorem \ref{theorem:DFP-H-robustness}}
\label{appendix:proof-V-DFP-H}
Let $H(\varepsilon)$ be the optimal solution of \eqref{eqn:perturbed-V-DFP-H}. 
Under the inexact line search, the influence function $\dot{H}(0)$ for $V$-DFP-H is equal
to $\Delta[H_k;y,\bar{y},s,s]$ which is defined in Lemma \ref{lemma:asympto-Beps}. 

First, we study the case that $\beta(z)$ is not the null function. 
Suppose $\beta(d)\neq 0$ for $d>0$. If $H_k$ satisfies $H_ky=s$, we have $H_k=H(0)$. Using
Lemma \ref{lemma:asympto-Beps} for the matrix $H_k$ such that $H_ky=s$, we obtain
\begin{align*}
 &\phantom{=}\dot{H}(0)\\
 &=
 \Delta[H_k;y,\bar{y},s,s]\\
 &=
 \frac{(y-\bar{y})^\top s}{s^\top y}
 \frac{\beta(\det{H_k})}{1-(n-1)\beta(\det{H_k})}
 \bigg[H_k-\frac{ss^\top}{s^\top y}\bigg]
 -\frac{H_k\bar{y}s^\top +s\bar{y}^\top H_k}{s^\top y}
 +\frac{(y+\bar{y})^\top s}{(s^\top y)^2} ss^\top. 
\end{align*}
The above expression is almost same as \eqref{eqn:perturb-BFGS-H-beta_nonzero}, and thus
the same proof remains valid to obtain
\begin{align*}
 \sup\big\{
 \|\dot{H}(0)\|_F~|~H_k\in\mathrm{PD}(n),\,\bar{y}\in{\mathcal{Y}}
 \big\}=\infty. 
\end{align*}

Next, we consider the case that $\beta$ is the null function. 
Then $V(z)=-\log(z)$ and $\nu(z)=1$ hold. 
For $H_k$ such that $H_ky=s$, we have
\begin{align*}
 \dot{H}(0)
 =\Delta[H_k;y,\bar{y},s,s]
 =-\frac{H_k\bar{y}s^\top +s\bar{y}^\top H_k}{s^\top y}
 +\frac{(y+\bar{y})^\top s}{(s^\top y)^2}ss^\top. 
\end{align*}
This is the same as the influence function of
\eqref{eqn:standard-BFGS-H-influence-func}, and thus, we obtain
\begin{align*}
 \sup\big\{\|\dot{H}(0)\|_F~|~H_k\in\mathrm{PD}(n),\,
 \bar{y}\in{\mathcal{Y}}\big\}=\infty.
\end{align*}

\bibliographystyle{plain}

\end{document}